\newif\ifRSSB
\RSSBfalse

\ifRSSB
\documentclass{statsoc}
\usepackage{geometry}

\else

\documentclass[12pt]{article}

\fi

\usepackage[utf8]{inputenc}
\usepackage{xcolor}
\usepackage{amsthm}
\usepackage{amsmath}
\usepackage{amssymb}
\usepackage{graphicx}
\usepackage{float}
\usepackage{bm}
\usepackage{natbib}
\usepackage{comment}
\usepackage{yfonts}

\newtheorem{remark}{Remark}
\newtheorem{proposition}{Proposition}
\newtheorem{lemma}{Lemma}
\newtheorem{theorem}{Theorem}

\newcommand{\R}{\mathbb{R}}

\ifRSSB

\title[Reversible Jump PDMP]{Reversible Jump PDMP Samplers for Variable Selection}
\author{Augustin Chevalier, Paul Fearnhead and Matt Sutton\footnote{This research was supported by EPSRC grants EP/R018561 and EP/R034710.}}
\address{Department of Mathematics and Statistics, Lancaster University}
\else
\title{Reversible Jump PDMP Samplers for Variable Selection}
\author{Augustin Chevalier, Paul Fearnhead and Matt Sutton\\ Department of Mathematics and Statistics, Lancaster University \footnote{This research was supported by EPSRC grants EP/R018561 and EP/R034710.}}
\fi

\begin{document}

\maketitle

\begin{abstract}
A new class of Markov chain Monte Carlo (MCMC) algorithms, based on simulating piecewise deterministic Markov processes (PDMPs), have recently shown great promise: they are non-reversible, can mix better than standard MCMC algorithms, and can use subsampling ideas to speed up computation in big data scenarios. However, current PDMP samplers can only sample from posterior densities that are differentiable almost everywhere, which precludes their use for model choice. Motivated by variable selection problems, we show how to develop reversible jump PDMP samplers that can jointly explore the discrete space of models and the continuous space of parameters. Our framework is general: it takes any existing PDMP sampler, and adds two types of trans-dimensional moves that allow for the addition or removal of a variable from the model. We show how the rates of these trans-dimensional moves can be calculated so that the sampler has the correct invariant distribution. Simulations show that the new samplers can mix better than standard MCMC algorithms. Our empirical results show they are also more efficient than gradient-based samplers that avoid model choice through use of continuous spike-and-slab priors which replace a point mass at zero for each parameter with a density concentrated around zero.

\end{abstract}

\noindent%
\ifRSSB
\else
\fi

\section{Introduction}

There currently is much interest in developing MCMC algorithms based on simulating piecewise deterministic Markov processes (PDMPs). These are continuous time Markov processes that have deterministic dynamics between a set of event times, and the randomness in these processes only comes through the random event times and potentially random transitions at the events \cite[see][for an introduction to PDMPs]{PDMPDavis}. 

The idea of simulating PDMPs to sample from a target distribution of interest originated in statistical physics \cite[]{peters2012rejection,michel2014generalized}, but has recently been proposed as an alternative to standard MCMC to sample from posterior distributions in Bayesian Statistics, with such algorithms as the Bouncy Particle Sampler \cite[]{bouchard2018bouncy} and the ZigZag algorithm \cite[]{bierkens2017piecewise,bierkens2019zig} amongst others \cite[]{vanetti2017piecewise,markovic2018bouncy,wu2020coordinate,michel2020forward,bierkens2020boomerang}. See \cite{fearnhead2018piecewise} for an introduction to this area.

To sample from a density $\pi(\bm{x})$ current PDMP samplers introduce a velocity component, $\bm{v}$, of the same dimension as $\bm{x}$, and have deterministic dynamics that correspond to a constant velocity model. At the random events the velocity component changes. Algorithms differ in terms of the event rate and how the velocity changes at each event, but each has a simple recipe for choosing these so that the resulting PDMP  has $\pi(\bm{x})$ as its invariant distribution. These recipes depend on $\pi(\bm{x})$ through the gradient of $\log \pi (\bm{x})$, which importantly means that $\pi(\bm{x})$ only needs to be known up to proportionality, but also that $\pi(\bm{x})$ needs to be differentiable almost everywhere. The advantages of PDMP samplers are that they are non-reversible, and thus can mix more quickly than standard reversible MCMC algorithms \cite[]{diaconis2000analysis}, and, when sampling from posterior distributions, they can use a small sample of data points at each iteration whilst still targeting the true posterior distribution \cite[]{bierkens2019zig}.

However, the restriction to sampling from densities that are differentiable means that current PDMP samplers cannot be used in model choice problems. The aim of this paper is to address this limitation, with a particular motivation of PDMP samplers that can be used in variable selection problems that are common in, for example, linear regression and general linear regression.  We show how to design efficient PDMP samplers which allow movement between different models.

A simple way to implement PDMP samplers for variable selection problems is to use continuous spike-and-slab priors on the parameters \cite[]{ishwaran2005spike, George93}, which, rather than setting some parameters exactly to 0, have priors that place substantial mass close to 0. With such a prior, the resulting posterior density is differentiable, and existing PDMP samplers can be used. However such an approach has three disadvantages. First, under such a prior it can be hard to interpret the results as we do not formally get posterior probabilities on whether certain variables should be included in the model. Second, they introduce an extra tuning parameter to the prior which governs the shape of the spike of the component. Third, as we show in Section \ref{sec:PDMP}, using PDMP samplers to sample from the resulting posterior can be computationally inefficient: the samplers will need to simulate many events so that the parameters associated to variables that should not be in the model are kept close to 0. 

We demonstrate how to adapt existing PDMP samplers to variable selection problems. 
Specifically they evolve as the PDMP sampler when exploring the posterior associated with a given model, but with two additional events: if any parameter value hits 0 the PDMP jumps to the smaller model where the corresponding variable is removed; whilst with some rate there are events that re-introduce variables into the model. We show in Section \ref{sec:RJPDMP} how to calculate the rate and transition for these new types of event so that the sampler has the correct invariant distribution. To calculate these we need different techniques than those used for existing PDMP samplers, as we need to account for the behaviour of the process when parameters hit zero. The techniques we use are most similar to those in \cite{bierkens2018piecewise}, which considers PDMPs with restricted domains. However in that paper the dynamics at the boundary of the domain could be chosen so that the net flow of probability at the boundary is zero; whereas we need to balance the probability flow out of a model which occurs when a parameter hits zero with the flow into the model caused by the events that re-introduce variables.

The approach we present is generic, in that it can take any current PDMP sampler and be used to obtain a version that can be applied to the variable selection problem. We call the new class of samplers reversible jump PDMP samplers, due to the analogy with reversible jump MCMC \cite[]{green1995reversible}. We show how to derive reversible-jump versions of both ZigZag and the Bouncy Particle Sampler in Section \ref{sec:exampleSamplers}, before investigating empirically these algorithms on both logistic regression and robust linear regression models. 

Proofs of all theorems are relegated to the appendix. Code for implementing the new reversible jump PDMP samplers, and for replicating our examples, is available from {\texttt{https://github.com/matt-sutton/rjpdmp}}.

\section{PDMP Samplers and Model Choice}
\label{sec:pdmp-model-choice}

\subsection{Variable Selection}

We will consider model selection problems that arise from variable selection. The general framework is that we have a vector of parameters, $\theta=(\theta_1, ...,\theta_p)$, and each model is characterised by setting some subset of the $\theta_j$s to 0. This is a common setting across linear models, generalised linear models and various extensions. 

To make ideas concrete, consider a linear model
$$
\bm{Y} = \sum_{j=1}^p\bm{X}_j\theta_j + \epsilon
$$
where $\bm{Y}$ is a vector of response variables, each $\bm{X}_j$ is a vector of covariates, and $\epsilon$ is an additive noise vector. When $p$ is large it is common to fit such a model under a sparsity assumption, namely that many of the $\theta_j$s are 0.

In a Bayesian analysis, such a sparsity assumption is encapsulated in our choice of prior on $\bm{\theta}$. To aid interpretation of the variable selection priors it is common to introduce a latent variable $\bm{\gamma} = (\gamma_1,...,\gamma_p)'$ where $\gamma_j = 1$ if the covariate $X_j$ is included in the model, i.e. if $\theta_j\neq0$. We let $|\bm{\gamma}|=\sum \gamma_j$ the number of covariates included in the corresponding model. Indexing $\bm{\theta}_{\bm{\gamma}}$ as the sub-vector of $\bm{\theta}$ with only the selected variables, any prior can be written in a hierarchical form where we have a prior on $\bm{\gamma}$, then conditional on $\bm{\gamma}$ we have a prior on $\bm{\theta}_{\bm{\gamma}}$, and set all remaining entries of $\bm{\theta}$ to 0. 

A special case is where each component of $\theta$ is independent of the others. In which case the prior can be written as
\begin{align*}
\theta_j \sim w_j g_j(\theta_j) + (1-w_j)\delta_0(\theta_j), \qquad j = 1,...,p,
\end{align*}
where $w_j \in (0,1)$ is the prior probability that $\gamma_j=1$, $\delta_0$ is a Dirac measure at zero and $g_j(\theta_j)$ is a distribution that models our prior beliefs for $\theta_j$ conditional on that variable being included in the model. Bayesian approaches to variable selection that put a probability mass on $\theta_j = 0$ in this way will be referred to as Dirac spike and slab methods. Notable examples of these methods include \cite{Mitchell88, Kuo98, Geweke96, Smith96, bottolo2010}.

While this formulation is natural from a modelling perspective it, sampling from the resulting posterior distribution can be challenging, with, for example, MCMC samplers that use gradient information such as Hamiltonian Monte Carlo \cite[]{neal2011mcmc} not being applicable.
To circumvent this issue it is common to use an approximation to this prior which replaces the point mass at 0 with a density that is peaked around 0, such as
\begin{align} \label{eq:ctsprior}
\theta_j \sim w_j \mathcal{N}(0,\tau_j^2) +  (1-w_j) \mathcal{N}(0,\tau_j^2c_j^2), \qquad j = 1,...,p,
\end{align}
where $c_j$ is taken small so that $\mathcal{N}(0,c_j^2\tau_j^2)$ approximates the Dirac spike. We will refer to Bayesian variable selection methods that replace the Dirac in the prior with a continuous approximation as continuous spike and slab methods. This prior was originally proposed in linear regression where it is commonly referred to as the stochastic search variable selection procedure \cite[]{George93}.

\subsection{PDMP Samplers} \label{sec:PDMP}

Piecewise Deterministic Markov Processes (PDMPs) are an emerging class of non-reversible continuous-time samplers. To be consistent with commonly used notation for PDMP samplers, we will consider sampling from a distribution with density $\pi(\bm{x})$ defined on some space $\mathcal{X}$; which is a slight change of notation relative to our regression model of the previous section. Current samplers augment the state to include a velocity vector and sample from a distribution on $E = \mathcal{X}\times\mathcal{V}$. In the following, for $\bm{z}\in E$ we will use the notation $\bm{z}=(\bm{x},\bm{v})$ with $\bm{x}\in\mathcal{X}$ a position and $\bm{v}\in\mathcal{V}$ a velocity. 

The PDMP can be defined by (i) deterministic dynamics between a set of random event times; (ii) the state-dependent rate at which events occur, $\lambda(\bm{z})$; and (iii) a probability distribution for the change in state at each event, with density $q(\bm{z}'|\bm{z})$. We will consider PDMPs whose deterministic trajectories follow a differential equation of the form:
\[
    \frac{\mbox{d}(\bm{x}_t,\bm{v}_t)}{\mbox{d}t} = (\bm{v}_t,\Phi(\bm{x}_t,\bm{v}_t))
\]
with $\Phi: E \rightarrow \mathcal{V}$ a smooth function. This setting contains the usual PDMP samplers such as ZigZag \cite[]{bierkens2019zig}, the Bouncy Particle Sampler \cite[]{bouchard2018bouncy}, or the Coordinate Sampler \cite[]{wu2020coordinate}. 

For example the ZigZag algorithm to simulate from $\pi(\bm{x})$ for $\bm{x}\in \R^p$, would introduce a velocity vector $\bm{v} \in \{-1,1\}^p$, and deterministic dynamics
\[
\frac{\mbox{d}(\bm{x}_t,\bm{v}_t)}{\mbox{d}t} = (\bm{v}_t,0),
\]
which are the dynamics of constant velocity model. Events occur at a rate that depends on the gradient of $\pi(\bm{x})$ in each component of the velocity, and at an event one component of the velocity is switched. The rate at which the $j$th component of the velocity, $\bm{v}^j_t$ is switched is just
\[
\max\left\{0, -\bm{v}^j_t \frac{\partial}{\partial \bm{x}^j} \log \pi(\bm{x}_t) \right\}.
\]
These rates depend on the target distribution just through the gradient of the log of the target -- which importantly means that we need only know the target distribution up to proportionality. 

We can apply current PDMPs, such as ZigZag, to the Bayesian variable selection problem if we use the continuous spike-and-slab prior (\ref{eq:ctsprior}). Realisations of such a sampler are shown in the left two plots of Figure \ref{fig:my_label} as we vary how concentrated the spike distribution is. 
For the more concentrated case, the sampler becomes inefficient as it involves many switching events when the state variable is close to 0. 

Intuitively as we make the variance of the spike component of the prior tend to 0 the prior converges to a prior with a point mass at 0; furthermore we can observe the output of our PDMP sampler ``converging" to a process shown in the right-hand plot of Figure \ref{fig:my_label}. Here, rather than the state having periods where it oscillates around 0, it has periods of time where $\theta_j=0$ and thus has many fewer events. 


\begin{figure}[t]
    \centering
    \includegraphics[width = \textwidth]{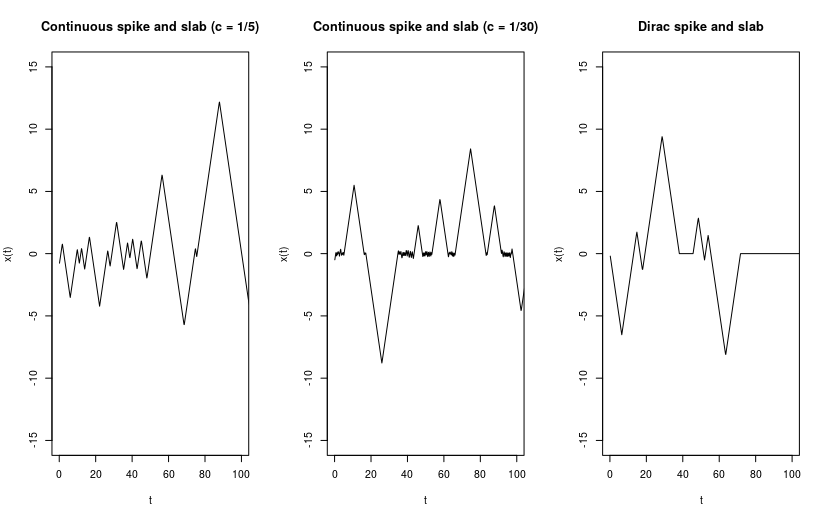}
    \caption{Sample paths of PDMPs implementing the variable selection priors in 1 dimension. The left and centre plots show the trajectories for a continuous spike-and-slab prior $0.5\mathcal{N}(0,\tau^2) + 0.5\mathcal{N}(0,\tau^2c^2)$ where $\tau^2 =16$. As $c$ decreases the spike component in the mixture approaches a Dirac mass. The figure on the right is the limiting process where we set the velocity to zero allowing the variable to stay fixed at zero.}
    \label{fig:my_label}
\end{figure}

\section{Reversible Jump PDMP Samplers} \label{sec:RJPDMP}

Let $\mathcal{M} = (\mathcal{M}_1,\mathcal{M}_2,...)$ be a set of models indexed by a parameter $k$. To each model $\mathcal{M}_k$ correspond a state space $\mathcal{X}^k$ of dimension $p_k$. For the sake of clarity, we will limit ourselves to variable selection, though it is straightforward to apply the results below to more general model selection problems. In our case $\mathcal{X}^k$ has a specific form: \[
    \mathcal{X}^k = \prod_i \mathbb{R}^{\gamma^k_i}
\]
where we abuse notation by using $\mathbb{R}^0 = \{0\}$ and where $(\gamma^k_i)_i$ is a sequence of numbers in $\{0,1\}$ representing whether a variable is enabled for model $\mathcal{M}_k$.
Let $\pi$ be the target posterior probability defined on $\mathcal{X} = \cup_k \mathcal{X}^k$. We further assume that the restriction of $\pi$ to each $\mathcal{X}^k$ has a density, and denote this by $\pi_k(\bm{x})$.
The first ingredient of our sampler is a collection of PDMPs defined for each model. Each PDMP sampler adds a velocity space $\mathcal{V}^k$ to the space $\mathcal{X}^k$ and samples from the space $E_k = \mathcal{X}^k\times\mathcal{V}^k$. 
Finally, for each model $\mathcal{M}_k$, the associated PDMP sampler has an extended infinitesimal generator $\mathcal{A}_k$ \cite[]{PDMPDavis} with invariant distribution proportional to $\nu_k$ with 
\[
\nu_k(\bm{x},\bm{v})=\pi_k(\bm{x}) p_k(\bm{v}|\bm{x}) \qquad (\bm{x},\bm{v}) \in \mathcal{X}^k\times\mathcal{V}^k,
\]
for some set of conditional densities $p_k(\bm{v}|\bm{x})$.

\begin{remark}
    For samplers such as ZigZag, $p_k(\mbox{d}\bm{v}|\bm{x})$ is a measure with support on a discrete set. By choosing $\mathcal{V}^k$ to be the support of $p_k(\mbox{d}\bm{v}|\bm{x})$, $p_k(\mbox{d}\bm{v}|\bm{x})$ still has a density -- and integrals can be interpreted as sums over the support points. This allows us to treat all samplers within the same framework.
\end{remark}

The second ingredient of our reversible jump PDMP sampler is a set of jumps between models. For our case of variable selection, we only allow adding or removing one variable at a time. Hence, let 
\[
\mathcal{T} = \{(i,j) \,|\, \sum_l |\gamma^i_l - \gamma^j_l| = 1 \text{ and } |\gamma^i| > |\gamma^j|\}
\]
be a set of pairs of transitions between models, with these ordered $(i,j)$ such that model $\mathcal{M}_j$ being obtained from $\mathcal{M}_i$ by removing one of the variables in $\mathcal{M}_j$. 
For transition $i \rightarrow j$, we define an active boundary $\Gamma_{i,j} = \mathcal{X}^j \times \mathcal{V}^i$, a subspace of $E_i$. Each trajectory passing through $\Gamma_{i,j}$ has some probability $p_{i,j}$ of jumping to $E_j$ using a deterministic jump function $g_{i,j}$. We assume that the jump function does not change the position. So, if a trajectory $\bm{z}_{t}$ of our process has a left limit at time $t$, $\bm{z}_{t-}$ in $\Gamma_{i,j}$, then with some probability $p_{i,j}$, $\bm{z}_t=g_{i,j}(\bm{z}_{t-})\in E_j$ with $\bm{x}_t = \bm{x}_{t-}$.
For transition $j \rightarrow i$, we introduce $\beta_{i,j}(\bm{z})$ a Poisson rate and a jump kernel $Q_{i,j}(\cdot,\bm{z})$ such that if the trajectory is in $E_j$, then with rate $\beta_{i,j}(\bm{z}_{t-})$, $\bm{z}_t$ is drawn from $Q_{i,j}(\cdot,\bm{z}_{t-})\in \Gamma_{i,j}$.  We impose symmetry in the jumps between models so that for any $\bm{z}' \sim Q_{i,j}(\cdot,\bm{z})$, $g_{i,j}(\bm{z}') = \bm{z}$.

\subsection{Invariant distribution}

Let $\nu = \sum_k \nu_k$ be a measure on $\cup_k E_k$. The $\bm{x}$-marginal distribution of $\nu$ is $\pi$ and this section provides conditions on $Q_{i,j}, p_{i,j}$ and  $\beta_{i,j}$ that has to be satisfied if $\nu$ is to be the invariant distribution of the process. To avoid additional technicalities that complicate the proof, we gives results for PDMPs with bounded velocity spaces.

\begin{theorem}
\label{th:extended-gen}
Assume $\mathcal{V}^k$ is bounded for each $k$, and that we can bound the event rate over any compact region; i.e if $K\in E$ is compact then $\max_{\bm{z}\in K} \lambda(\bm{z})<\infty$. 
Suppose a measure $\mu$ is a stationary distribution with density on each $E_i$. Then for each function $f$ in a suitable set $\mathcal{F}$:
\begin{align}
\begin{split}
0
    &= \sum_{i} \int_{E_i} \mathcal{A}_i f \mbox{d}\mu 
    + \sum_{(i,j) \in \mathcal{T}} p_{i,j} \int_{\Gamma_{i,j}} \left[ f(g_{i,j}(\bm{z})) - f(\bm{z}) \right] \mu(\bm{z}) |\langle \bm{v},\bm{n}_{i,j}\rangle| \, \mbox{d}\bm{z}\\
    &+ \sum_{(i,j) \in \mathcal{T}} \int_{E_j} \beta_{i,j}(\bm{z}) \left( \int_{g_{i,j}^{-1}(\{\bm{z}\})} [f(\bm{z}') - f(\bm{z})] Q_{i,j}(\bm{z}'|\bm{z}) \, \mbox{d}\bm{z}'\right) \mu(\bm{z}) \mbox{d}\bm{z},
    \end{split}
    \label{eq:th1}
\end{align}
where $\bm{n}_{i,j}$ is the normal to $\Gamma_{i,j}$.\\
\end{theorem}
The set $\mathcal{F}$ in this theorem is the domain of the generator of our PDMP, and is described precisely in \cite{PDMPDavis}.

To get a directly usable conditions on $\beta_{i,j}, p_{i,j}$ and $Q_{i,j}$,  some additional notation must be introduced.
When jumping from $\bm{z}=(\bm{x},\bm{v}) \in E_j$ to $\bm{z}'=(\bm{x},\bm{v}') \in E_i$, the dimension of the velocity vector needs to be increased by one, but the position is unchanged. Hence $g_{i,j}^{-1}(\bm{z})$ is a one-dimensional manifold, which can be related to a subset, $U$ say, of $\mathbb{R}$ by introducing a function
\[
G_{i,j}: U \times E_j \rightarrow E_i
\]
such that for any $\alpha \in U$, we have $G_{i,j}(\alpha,\bm{z}) \in g_{i,j}^{-1}(\{\bm{z}\})$ and for a fixed $\bm{z} \in E_j$, $\alpha \mapsto G_{i,j}(\alpha,\bm{z})$ is a one to one mapping from $U$ to $g_{i,j}^{-1}(\bm{z})$ \cite[similar ideas are seen in reversible jump MCMC; see][]{green1995reversible}. 
For a given $\bm{z} \in E_j$, it is natural to rewrite the jump kernel $Q_{i,j}$ in terms of $\alpha \in U$, and henceforth we abuse notation and write $Q_{i,j}(\cdot,\bm{z})$ as a density on $U$. That is we can simulate the transition from $E_j$ to $E_i$, by simulating $\alpha \sim Q_{i,j}(\cdot,\bm{z})$ and then setting $\bm{z}'=G_{i,j}(\alpha,\bm{z})$.

Finally, let 
\[
    \nu_{i,j}(\bm{z}) = \nu(\bm{z}) |\langle \bm{v},\bm{n}_{i,j} \rangle|
\]
be an unnormalised density on $\Gamma_{i,j}$ and let $\bar{\nu}_{i,j}$ be the pushforward measure on $E_j$ of $\nu_{i,j}$ by $g_{i,j}$:
\[
    \bar{\nu}_{i,j}(B) = \int_{g_{i,j}^{-1}(B)} \nu_{i,j}(\bm{z}) \mbox{d}\bm{z} \qquad \mbox{ for any } B \subset E_j \mbox{ measurable.}
\]
Informally this is the measure under $\nu_{i,j}$ associated with values of $\bm{x}\in \Gamma_{i,j}$ that would be mapped by $g_{i,j}$ to the set $B$.

We are now in a position to give simple conditions for (\ref{eq:th1}) of Theorem \ref{th:extended-gen} to hold. To do this it is helpful to consider separately the cases where the space of velocities is continuous and the case where it is discrete.
\begin{theorem}
\label{th:invariant}
Assume the space of velocities is continuous.
A sufficient condition for (\ref{eq:th1}) of Theorem \ref{th:extended-gen} to hold for $\nu$ is that, for every $(i,j) \in \mathcal{T}$, the following conditions holds
\begin{align}
    \beta_{i,j}(\bm{z}) &= p_{i,j} \frac{\bar{\nu}_{i,j}(\bm{z})}{\nu(\bm{z}) } \qquad \text{ for } \bm{z}\in E_j \label{eq:thm2.1} \\
    Q_{i,j}(\alpha|\bm{z}) &= \frac{\nu_{i,j}(G_{i,j}(\alpha,\bm{z})) |J_{G_{i,j}}(\alpha,\bm{z})|}{\bar{\nu}_{i,j}(\bm{z})} \qquad \text{ for } \bm{z}\in E_j \text{ and } \alpha \in U \label{eq:thm2.2},
\end{align}
where $J_{G_{i,j}}$ denotes the Jacobian associated with the transformation $G_{i,j}$.
\end{theorem}
Intuitively, this result can be understood as a detailed balance condition: we balance the probability flow for each jump $\bm{z}\rightarrow \bm{z}'$ from $E_i$ to $E_j$, with that of the reverse jump $\bm{z}'\rightarrow \bm{z}$.

For discrete velocity spaces the result is slightly simpler, as the Jacobian term is not required.
\begin{theorem}
\label{th:invariantDiscrete}
Assume the velocity space is discrete. A sufficient condition for (\ref{eq:th1}) of Theorem \ref{th:extended-gen} to hold for $\nu$ is that, for every $(i,j) \in \mathcal{T}$, the following conditions holds
\begin{align}
    \beta_{i,j}(\bm{z}) &= p_{i,j} \frac{\bar{\nu}_{i,j}(\bm{z})}{\nu(\bm{z}) } \qquad \text{ for } \bm{z}\in E_j \label{eq:thm2.3} \\
    Q_{i,j}(\alpha|\bm{z}) &= \frac{\nu_{i,j}(G_{i,j}(\alpha,\bm{z})) }{\bar{\nu}_{i,j}(\bm{z})} \qquad \text{ for } \bm{z}\in E_j \text{ and } \alpha \in U \label{eq:thm2.4}.
\end{align}
\end{theorem}

\section{Example Samplers} \label{sec:exampleSamplers}

\subsection{ZigZag Sampler}

We first derive the jump rates and transitions for the ZigZag sampler described in Section \ref{sec:PDMP}. For this sample the velocity of each active component is either 1 or -1, which simplifies the computation.

We choose $g_{i,j}$ to be the projection that sets to $0$ the disabled variable.
Let $(i,j)\in \mathcal{T}$ be a transition. For any $\bm{v}\in \mathcal{V}^i$, we have $|\langle\bm{v},\bm{n}_{i,j}\rangle| = 1$. Therefore on $E_i$, 
\[
    \nu_{i,j}(\bm{z}) = \nu(\bm{z}) = \pi_i(\bm{x}) 2^{-|\gamma_i|}.
\]
For $\bm{z}\in E_j$, since a velocity component in $\{-1,1\}$ is projected to $0$,
\[
    \bar{\nu}_{i,j}(\bm{z}) = 2 \pi_i(\bm{x}) 2^{-|\gamma_i|} = \pi_i(\bm{x})  2^{-|\gamma_j|}.
\]
For the jump to $E_i$ we change the velocity for the component of the model that is added, whilst the other components of the velocity is unchanged. 
Denoting the new velocity of the added component by $\alpha$, we have from Theorem \ref{th:invariantDiscrete} that
\begin{align*}
    \beta_{i,j} &= p_{i,j} \frac{\pi_i(\bm{x})}{\pi_j(\bm{x})}\\
    Q_{i,j}(\alpha|\bm{z}) &= 1/2 \quad \text{for } \alpha \in \{-1,1\}.
\end{align*}
For our variable selection problem, the ratio of the posterior density that appears in $\beta_{i,j}$ will simplify to the ratio of the priors as the likelihood terms are common and cancel. If we have independent priors on the parameters for each variable this term will be a constant, which simplifies the simulation of the events at which we add new variables into our model. This comment applies also to the rates for the Bouncy Particle Sampler which we derive next. 

\subsection{Bouncy Particle Sampler}

We consider two versions of the Bouncy Particle Sampler. The first version has velocities on the unit sphere, so 
\[
    \mathcal{V}^i = \{\bm{v} \in \mathbb{R}^{|\gamma^k|} \, \text{ such that } \, \|\bm{v}\| = 1\}.
\]

Like the ZigZag sampler, the deterministic dynamics are given by a constant velocity model. The event rate for sampling from a density $\pi_k(\bm{x})$ is
\[
\lambda(\bm{z})=\max\{0, -\bm{v} \cdot \nabla_{\bm{x}} \log \pi_k(\bm{x}) \},
\]
with the velocity reflecting in the normal to $\log\pi_k(\bm{x})$ at an event. The Bouncy particle also often has refresh events, at which a completely new velocity is sampled.

Extending the Bouncy Particle Sampler to the variable selection problem requires a more careful analysis than for the ZigZag sampler due to the geometry of the velocity space. For the jump to $E_i$, we will let $\alpha$ denote the new velocity for the component added to the model. To ensure the new velocity lies on the unit sphere, we re-scale the other components of the velocity: so if the old velocity is $\bm{v}$ then the new velocity is $\alpha \bm{n}_{i,j}+\sqrt{1-\alpha^2}\bm{v}$. The following proposition states how to choose the jump rate, $\beta_{i,j}$ and the density for $\alpha$.

\begin{proposition} \label{prop:1}
For the Bouncy Particle Sampler with velocities on the unit sphere, (\ref{eq:thm2.1}) and (\ref{eq:thm2.2}) 
are satisfied if, for all $(i,j)\in \mathcal{T}$ with  $|\gamma^j| > 0$:
\begin{align}
    \beta_{i,j} &= p_{i,j} \frac{\pi_i(\bm{x})}{\pi_j(\bm{x})}  \frac{2 A_{sphere}(|\gamma^j|)}{A_{sphere}(|\gamma^i|)} \frac{1}{|\gamma^j|}\\
    Q_{i,j}(\alpha|\bm{z}) &= \frac{|\alpha| |\gamma^j| \sqrt{1-\alpha^2}^{|\gamma^j|-2}}{2} \text { for } \bm{z} \in E_j \text{ and } \alpha \in (-1,1)
\end{align}
with $A_{sphere}(|\gamma^i|) = \frac{\Gamma(\frac{|\gamma^i|}{2})}{2\pi^{\frac{|\gamma^i|}{2}}}$ the area of the unit sphere of $\mathbb{R}^{|\gamma^i|}$; and if for $|\gamma^j| = 0$, where the Bouncy Particle Sampler and ZigZag are equivalent, we use the ZigZag rates.
\end{proposition}

The second version of the Bouncy Particle Sampler has velocities in $\mathbb{R}^{|\gamma^k|}$, with their density being standard Gaussian and independent of $\bm{x}$. The dynamics are as previously. This sampler does not satisfy our requirement that the velocity space is bounded, though the argument behind Theorem \ref{th:extended-gen} can be extended to this sampler, and our results would apply to this sampler if we replaced the Gaussian density by a truncated Gaussian (with any arbitrarily large truncation value).

For the jump to $E_i$, we again let $\alpha$ denote the new velocity for the component added to the model, whilst for this version the other components of the velocity remain unchanged.

\begin{proposition} \label{prop:2}
For the Bouncy Particle Sampler with Gaussian velocities, (\ref{eq:thm2.1}) and (\ref{eq:thm2.2}) 
are satisfied if, for all $(i,j)\in \mathcal{T}$: 
\begin{align}
    \beta_{i,j} &= p_{i,j} \frac{\pi_i(\bm{x})}{\pi_j(\bm{x})}  \frac{2}{\sqrt{2\pi}} \\
    Q_{i,j}(\alpha|\bm{z}) &= 2 |\alpha| e^{-\frac{1}{2}\alpha^2} \text { for } \bm{z} \in E_j \text{ and } \alpha \in \mathbb{R}.
\end{align}
\end{proposition}

\section{Simulation Study}

In this section we demonstrate the potential advantage of our new samplers compared to alternative approaches for Bayesian variable selection. To compare between different samplers we consider the Monte Carlo estimates of the posterior probabilities of inclusion, the posterior means for the regression coefficients, and the posterior means conditioned on the model. Similar comparisons methods are used in \cite{Zanella19} for comparing efficiency of MCMC methods on Bayesian variable selection problems. For a given sampler the statistical efficiency is measured by the mean squared error of the sampler denoted by $\sigma_{\text{sampler}}^2$ and is calculated using $R$ runs of the sampler as 
\begin{equation}
\sigma_{\text{sampler}}^2 = \frac{1}{R}\sum_{i=1}^R(\hat{q}_r - q)^2
\label{eq:mse}
\end{equation}
where $\hat{q}_r$ for $r=1,...,R$ are the estimates of a quantity of interest from the $R$ runs, and $q$ is either the exact posterior quantity of interest, if available, otherwise it is the estimate from an independent long run of an MCMC method. In multiple dimensions the statistical efficiency is measured as the median $\sigma^2_{\text{sampler}}$ over all dimensions. To compare the performance of different samplers we also consider a measure of efficiency relative to a reference sampler. If we denote the reference sampler by ref, then we define Relative Statistical Efficiency (RSE) and Relative Efficiency (RE) by
$$
\text{RSE} = \frac{\sigma_{\text{ref}}^2~n_{\text{ref}}}{\sigma_{\text{sampler}}^2~n_{\text{sampler}}}, \qquad \text{RE} = \frac{\sigma_{\text{ref}}^2~t_{\text{ref}}}{\sigma_{\text{sampler}}^2~t_{\text{sampler}}},
$$
where $n_{\text{sampler}}$ and $n_{\text{ref}}$ are the number of iterations of the algorithms and $t_{\text{sampler}}$ and $t_{\text{ref}}$ are the computation times of the algorithms. The RSE measures the relative efficiency of the algorithms per iteration whereas the RE measures the efficiency per second. For interpretation an RSE or RE value of 2 implies that the sampler is 2 times \textit{more efficient} than the reference method. The sensitivity of the methods to the choice of reversible jump parameters $p_{i,j}$ and regular PDMP tuning parameters is explored empirically on a simple model in the Supplementary Material. Based on these results we fixed $p_{i,j} = 0.6$ for all $i$ and $j$ in all reversible jump PDMP samplers and fixed the refreshment rate at 0.1 for the reversible jump BPS methods for the following results.

\subsection{Logistic regression}
First we compare PDMP based samplers with a collapsed Gibbs sampler and a reversible jump sampler on a logistic regression problem. The Gibbs sampler is based on the Polya-Gamma sampler of \cite{polson2013bayesian}: details of this sampler are given in the Supplementary Material. We implemented reversible jump MCMC using the NIMBLE software package \cite[]{NIMBLE, NIMBLER} using independent Normal proposal for selected variables. 
\\
The logistic regression model has a $p$-dimensional regression parameter $\bm{\theta} \in \mathbb{R}^p$, and a binary response $y_i \in \{0,1\}$ which is distributed as
\begin{align*}
P(y_i = 1 | \bm{x}_i, \bm{\theta}) = \frac{\exp(\sum_{j=1}^px_{i,j}\theta_j)}{1+\exp(\sum_{j=1}^px_{i,j}\theta_j)}
\end{align*}
where $\bm{x}_i$ is the $p$-vector of covariates for observation $i$. In our simulation study, each vector $\bm{x}_i$ is simulated from a multivariate Normal with mean zero and $p\times p$ covariance matrix $\Sigma$. We use a prior that is independent for each $\theta_j$ and where $\theta_j \sim \frac{10}{p}\mathcal{N}(0, 10) +  (1-\frac{10}{p})\delta_0$, which corresponds to a prior that favors models with 10 selected variables. Data was generated using this model and the following choices for $\bm{\theta}$ and covariance matrix $\Sigma$:

\begin{enumerate}
\item A pair of correlated covariates, one of which is in the model: $\bm{\theta} = (1,0,...0)^T$ with $\Sigma_{2,1}=\Sigma_{2,1} = 0.9$, with $\Sigma_{i,i} = 1$ and $\Sigma_{i,j} = 0$ if $i\neq j$.
\item Structured correlation between all covariates with six active covariates: \\ $\bm{\theta} = (3,3,-2,3,3,-2,0,...,0)^T$ with $\Sigma_{i,j} = \exp(-|i-j|)$.
\item No correlation between covariates and six active covariates:\\ $\bm{\theta} = (3,3,-2,3,3,-2,0,...,0)^T$,  with $\Sigma_{i,i} = 1$ and $\Sigma_{i,j} = 0$ if $i\neq j$. 
\end{enumerate}

These simulation scenarios are analogous to others previously considered in the literature for linear regression. Scenarios 1 and 3 are are similar to considered by \cite{wang2011} and \cite{Zanella19} while Scenario 2 is similar to one considered by \cite{yang2016}. We present results for both ZigZag and the Bouncy Particle Sampler with Gaussian distributed velocities (almost identical results are obtained for the Bouncy Particle Sampler where the velocities are on the unit sphere). See Section \ref{supp:RE} of the Appendix for more details.

The PDMP methods are competitive with Gibbs sampling in low sample sizes and offer substantial efficiency gains for larger sample sizes. Smaller gains can also be seen when the dimension is increased with fixed sample size. Both BPS and ZigZag methods offer similar relative efficiencies across the experiments. The greatest efficiency gains for reversible jump PDMP methods was seen in Scenario 1 with Scenarios 2 and 3 offering lower gains in performance. This may be due to smaller models being more likely in Scenario 1 and thus less computational effort required for the PDMP methods in gradient calculations and simulation of event times. For all experiments the Gibbs sampler offers a computational advantage over the reversible jump methods in relative efficiency for marginal posterior probabilities of inclusion. This increased performance is unsurprising as the model update steps in the Gibbs sampler have marginalised over the parameter values $\bm{\theta}$ yielding efficient moves through the model space. In more general settings where it is not possible to integrate over $\bm{\theta}$ the mixing will be substantially poorer. 

\begin{table}
\caption{{Scenario 1 (pair of correlated variables)}: Relative efficiencies (RE) for methods, against a Reversible Jump algorithm, for the marginal posterior means (Mean) and marginal posterior probabilities of inclusion (PI).}
\centering
\begin{tabular}{ |p{1.5cm}||p{1cm}|p{1cm}|p{1cm}|p{1cm}|p{1cm}|p{1cm}|  }
\hline
 & \multicolumn{2}{|c|}{ZigZag} & \multicolumn{2}{|c|}{BPS} & \multicolumn{2}{|c|}{Gibbs}\\
\hline\hline
$n$, $p$ & PI & Mean & PI & Mean & PI & Mean \\
\hline
  \hline
100, 100 & 1.75 & 3.11 & 1.92 & 2.28 & 2.12 & 4.25 \\ 
  200, 100 & 5.60 & 8.75 & 5.49 & 11.62 & 2.76 & 2.83 \\ 
  400, 100 & 14.70 & 25.08 & 12.47 & 29.07 & 3.66 & 4.14 \\ 
  800, 100 & 20.18 & 30.45 & 17.59 & 34.38 & 3.71 & 3.79 \\ 
  \hline
  100, 200 & 1.30 & 1.38 & 1.69 & 1.81 & 2.54 & 2.26 \\ 
  200, 200 & 14.92 & 20.12 & 14.85 & 26.02 & 2.92 & 2.64 \\ 
  400, 200 & 23.30 & 34.21 & 21.44 & 39.83 & 2.86 & 2.89 \\ 
  800, 200 & 43.98 & 59.56 & 36.38 & 62.89 & 3.14 & 3.05 \\ 
  \hline
  100, 400 & 1.17 & 1.76 & 2.56 & 4.07 & 1.62 & 2.15 \\ 
  200, 400 & 12.70 & 6.72 & 14.64 & 17.22 & 2.42 & 2.17 \\ 
  400, 400 & 43.73 & 62.07 & 34.13 & 61.47 & 2.41 & 2.43 \\ 
  800, 400 & 102.36 & 152.02 & 78.26 & 133.46 & 2.68 & 2.78 \\ 
\hline
\end{tabular}

\end{table}

\begin{table}
\caption{{Scenario 2 (General correlation)}: Relative efficiency (RE) for methods, against a Reversible Jump algorithm, for the marginal posterior means (Mean) and marginal posterior probabilities of inclusion (PI).}
\centering
\begin{tabular}{ |p{1.5cm}||p{1cm}|p{1cm}|p{1cm}|p{1cm}|p{1cm}|p{1cm}|  }
\hline
 & \multicolumn{2}{|c|}{ZigZag} & \multicolumn{2}{|c|}{BPS} & \multicolumn{2}{|c|}{Gibbs}\\
\hline\hline
$n$, $p$ & PI & Mean & PI & Mean & PI & Mean \\
\hline
  \hline
100, 100 & 0.57 & 0.83 & 1.15 & 1.86 & 1.05 & 0.76 \\ 
  200, 100 & 0.88 & 0.92 & 1.71 & 2.10 & 1.96 & 1.27 \\ 
  400, 100 & 1.39 & 1.50 & 2.11 & 2.32 & 1.69 & 1.00 \\ 
  800, 100 & 4.77 & 7.30 & 5.63 & 9.71 & 1.99 & 1.78 \\ 
  \hline
  100, 200 & 0.76 & 1.21 & 2.03 & 3.03 & 1.53 & 1.37 \\ 
  200, 200 & 2.28 & 2.80 & 4.87 & 6.84 & 1.92 & 1.31 \\ 
  400, 200 & 4.10 & 4.51 & 6.77 & 8.77 & 1.30 & 0.87 \\ 
  800, 200 & 10.36 & 13.23 & 11.12 & 19.21 & 1.41 & 1.12 \\ 
  \hline
  100, 400 & 1.09 & 1.27 & 2.93 & 4.01 & 1.12 & 0.88 \\ 
  200, 400 & 3.57 & 3.37 & 6.57 & 7.21 & 1.24 & 0.66 \\ 
  400, 400 & 8.05 & 8.65 & 13.45 & 16.24 & 1.15 & 0.83 \\ 
  800, 400 & 26.93 & 34.41 & 28.71 & 51.59 & 1.32 & 1.20 \\ 
\hline
\end{tabular}

\end{table}

\begin{table}
\caption{{Scenario 3 (No correlation)}: Relative efficiency (RE) for methods, against a Reversible Jump algorithm, for the marginal posterior means (Mean) and marginal posterior probabilities of inclusion (PI).}
\centering
\begin{tabular}{ |p{1.5cm}||p{1cm}|p{1cm}|p{1cm}|p{1cm}|p{1cm}|p{1cm}|  }
\hline
 & \multicolumn{2}{|c|}{ZigZag} & \multicolumn{2}{|c|}{BPS} & \multicolumn{2}{|c|}{Gibbs}\\
\hline\hline
$n$, $p$ & PI & Mean & PI & Mean & PI & Mean \\
\hline
  \hline
100, 100 & 0.64 & 0.87 & 1.25 & 1.91 & 1.19 & 1.04 \\ 
  200, 100 & 1.40 & 1.54 & 2.52 & 3.12 & 2.20 & 1.69 \\ 
  400, 100 & 2.74 & 2.79 & 3.07 & 3.67 & 2.22 & 1.54 \\ 
  800, 100 & 5.20 & 8.59 & 5.64 & 10.74 & 1.72 & 1.50 \\
  \hline
  100, 200 & 0.76 & 0.94 & 1.56 & 1.77 & 1.24 & 1.13 \\ 
  200, 200 & 2.04 & 1.63 & 4.00 & 3.18 & 2.07 & 1.23 \\ 
  400, 200 & 7.60 & 12.68 & 11.90 & 24.23 & 1.63 & 1.38 \\ 
  800, 200 & 14.26 & 21.65 & 16.64 & 28.09 & 1.56 & 1.32 \\ 
  \hline
  100, 400 & 1.03 & 1.44 & 2.57 & 4.17 & 1.17 & 0.97 \\ 
  200, 400 & 2.31 & 1.59 & 4.18 & 2.98 & 1.37 & 0.84 \\ 
  400, 400 & 9.38 & 10.42 & 14.66 & 18.60 & 1.40 & 1.07 \\ 
  800, 400 & 30.73 & 49.94 & 35.35 & 71.83 & 1.42 & 1.26 \\ 
\hline
\end{tabular}
\end{table}


The use of subsampling methods for Bayesian variable selection problems is a recently emerging area \cite[]{Qifan20, buchholz20} though it remains under-studied. One of the attractions of PDMP samplers is that they can be implemented in a way where they only access a small subset of data at each iteration, whilst still targeting the true posterior. We now investigate how these ideas work in the variable selection problem, by comparing the efficiency of three implementations of ZigZag with that of the Gibbs sampler, and see how this depends on the number of observations.  These are ZigZag using the full data, ZigZag using subsampling with a global bound, and ZigZag with subsampling control variates \cite[see][for details of both subsampling approaches]{bierkens2019zig}.


Standard application of control variates requires calculation of the gradient at a reference point using the full likelihood. Due to the trans-dimensional nature of variable selection problems, a full gradient calculation is not well defined. For this reason we choose to make use of control variates defined for a fixed model $\mathcal{M}$ where the gradient is well defined. These control variates are only used when the sampler is in this model. Extensions where multiple control variates are used for models with high probability would be straightforward. 

\begin{figure}[!h]
    \centering
    \includegraphics[width = \textwidth]{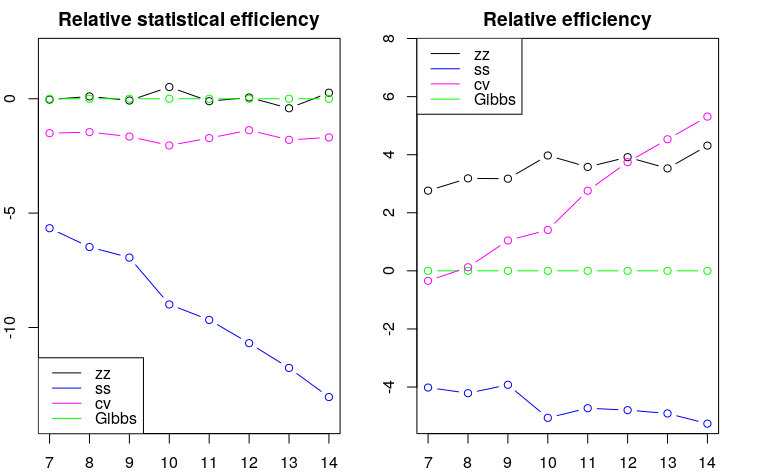}
    \caption{Log-log plots of efficiency, relative to the Gibbs sampler, of different samplers as we vary the number of observations. Plotted are the relative efficiencies for the posterior mean conditional on model $\mathcal{M}^*$ where $\mathcal{M}^*$ corresponds to the true data generated model. The dataset was generated with a 15-dimensional regression parameter $\bm{\theta} = (1,1,0,0,...,0)$. The methods run are the Zig-Zag applied to the full dataset (zz, black), Zig-Zag with subsampling using global bounds (ss, blue), Zig-Zag with control variates (cv, magenta) and Gibbs sampling (Gibbs, green). 
    All methods were initialised at the location of the control variate. Methods were given the same computational budget, for details see Section \ref{supp:RE} of the Appendix.}
    \label{fig:my_label1}
\end{figure}


Results are shown in Figure \ref{fig:my_label1}.
Despite our simplistic implementation, these results indicate that Zig-Zag with control variates is becoming increasingly efficient relative to Zig-Zag using the full dataset as the number of samples increases. 
Furthermore we see evidence of super-efficiency -- whilst the computational cost per ESS of the Gibbs sampler is expected to be linear in the number of observations, the relative efficiency plots suggest that this is sub-linear for ZigZag with control variates.

\subsection{Robust regression}

As mentioned in the introduction, a common approach to Bayesian variable selection is to use continuous spike-and-slab priors for each parameter rather than try to sample from the joint posterior of model and parameters. Such an approach is attractive as it enables standard gradient-based samplers, such as Hamiltonian Monte Carlo, to be used. We now compare such an approach, implemented with the popular STAN software \cite[]{carpenter2017stan}, to our PDMP samplers. Our aim is to both investigate the computational efficiencies of the two approaches and to show the differences in posterior that we obtain from these different types of prior. Our comparison is based on a robust linear regression model.

\begin{figure}[H]
    \centering
    \includegraphics[width = \textwidth]{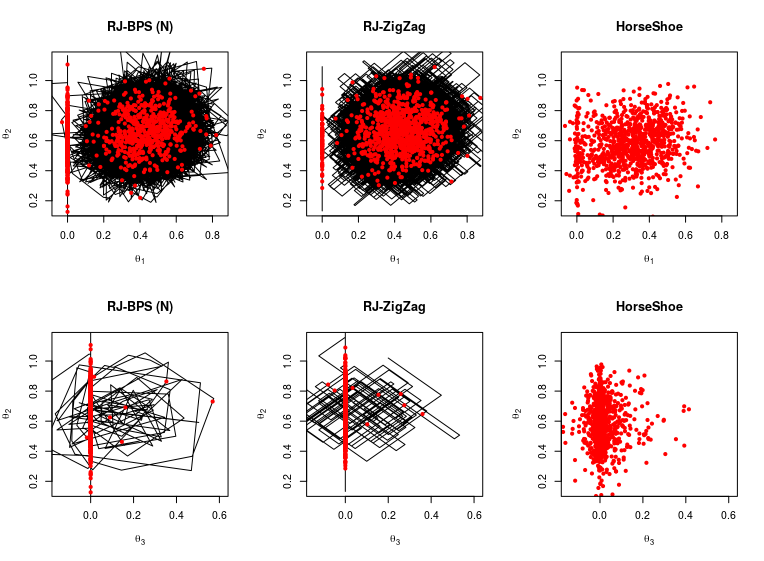}
    \caption{Dynamics of the samplers on a robust regression example with spike and slab or horseshoe prior. The top row shows the posterior for $\theta_1$ and $\theta_2$, bottom row shows the estimates for $\theta_2$ and $\theta_3$. The spike and slab distributions are sampled using the reversible jump PDMP samplers with reversible jump parameter $0.6$ and refreshment for the BPS methods set to 0.5. All methods are shown with $10^3$ samples (red) and the PDMP dynamics are shown in black. Sampling with the Horseshoe prior was implemented in Stan using NUTS. Both Stan and PDMP methods were run for the same computing time. To aid visualisation only the first $30\%$ of the PDMP trajectories are shown. }
\label{fig:dynamicsRR}
\end{figure}

In particular, we model the errors in our linear regression model as a mixture of Normals with different variances. Thus
\begin{align*}
\bm{Y} = \sum_{i=1}^p \bm{X}_j\theta_j + \epsilon, \qquad
\epsilon \sim~ \frac{1}{2}\text{N}(0, 1) + \frac{1}{2}\text{N}(0, 10^2).
\end{align*}
The continuous variable selection prior we consider is the regularised horseshoe \cite[]{piironen17a,piironen17}
\begin{align*}
\theta_j \sim N(0, \tau^2\tilde{\lambda}_j),\qquad 
\tilde{\lambda}_j = \frac{c^2\lambda_j}{c^2+\tau^2\lambda_j}, \qquad
\lambda_j \sim C^+(0,1)
\end{align*}
for $j=1,...,p$ where $C^+(0,1)$ denotes the half-Cauchy distribution for the standard deviation $\lambda_j$. 
The regularised horseshoe is a variation of the horseshoe prior \cite[]{carvalho2010horseshoe} that offers a continuous approximation of a Dirac spike and slab where the slab is a Normal distribution with finite variance $c^2$. The hyper-parameter $\tau$ controls the global shrinkage of the variables towards zero. In \cite{carvalho2010horseshoe} it was shown that for standard linear regression the optimal choice for a fixed value of $\tau$ is $\tau_0 = \sigma\frac{p_0}{(p-p_0)\sqrt{n}}$ where $p_0$ is a the number of nonzero variables in the sparse model and $\sigma$ is the noise variance. In line with this, we compare the regularised horseshoe prior with fixed hyper-parameter $\tau_0 = \frac{p_0}{(p-p_0)\sqrt{n}}$  against the spike-and-slab  prior 
$$\theta_j \sim \frac{p_0}{p}N(0,c^2) + \left(1-\frac{p_0}{p}\right)\delta_0$$ 
for $j=1,...,p$. MCMC for the model using a horseshoe prior was performed by Stan's implementation of NUTS \cite[]{hoffman2014no}.


We first compare the variable selection dynamics for a simple model with $p=4$ variables, $n = 120$ observations and regression parameter $\bm{\theta} = (0.5,0.5,0,0)^T$. The covariate values and residuals were generated as independent draws from a standard Normal and the prior expected model size is set to $p_0 = 1$. Example output for the PDMP samplers and the Stan implementation is shown in Figure \ref{fig:dynamicsRR}. The posteriors show the horseshoe prior replicating the effect of the spike-and-slab through shrinking the coefficients towards zero, but it is not able to give exact zeros. 

We now compare the reversible jump PDMP methods in terms of their sampling efficiency for a higher dimensional problem. The dataset is generated for $p=200$ variables and $n = 100$ observations with regression parameter $\bm{\theta} = (2,2,2,2,0,0,0....0)^T$. The covariates for each observation were drawn from an AR(1) process with lag-1 correlation of 0.5. The residuals were generated from a standard Cauchy distribution. 

We ran Stan with the default settings for a burnin of 1000 iterations and then for 16, 32, 64, ..., 2048 iterations. We ran the reversible jump PDMP samplers for the same wall clock time for both burnin and subsequent iterations. The experiment was repeated 35 times and the resulting boxplots of the posterior mean of $\theta_1$ are shown in Figure \ref{fig:samplingRR}. 
\begin{figure}[H]
    \centering
    \includegraphics[width = \textwidth]{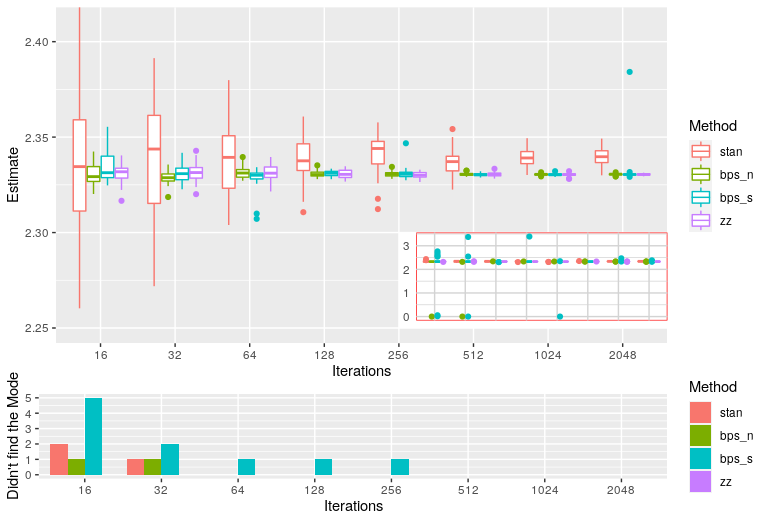}
    \caption{Sampling efficiency for reversible jump PDMP vs Stan for the robust regression example. The top figure is boxplots of the posterior mean of $\theta_1$ for increasing computational budget, with outliers from the sampler removed for visualisation purposes. These removed outliers correspond to times that the sampler has become stuck in a local mode where $\theta_1 = 0$.  The subplot shows the full results including outliers from the sampler. Both the Stan and reversible jump PDMP methods converge to similar solutions despite having different priors. The bottom figure shows the number of times that the samplers did not find the global mode.}
\label{fig:samplingRR}
\end{figure}
For the same computational budget, the reversible jump PDMP methods are able to attain better performance. However, it is also apparent from this simulation that the BPS reversible jump PDMP methods are more sensitive to local modes. It is unclear if this is an artifact of the sampler or due to the choice of refreshment rate $(0.6)$. This sensitivity is seen to diminish as the sampler is run for longer.  

The predictive ability of the methods is compared in Figure \ref{fig:predRR}. Here an additional $n=100$ observations were drawn from the same model and these were used as a hold-out dataset to validate the posterior predictive ability. The predictive ability is defined in terms of the Monte Carlo estimate of the mean square prediction performance $\frac{1}{100}\sum_{i=1}^{100}(y_i - \bm{x}_i^T\bm{\theta})^2$ where $\bm{\theta}$ is replaced by the samples generated by either Stan or samples given by a discretisation of time for the reversible jump PDMP samplers. The reversible jump PDMP samplers all give the same predictive performance for large iteration numbers while Stan, which uses a horseshoe prior, performs slightly worse.

\begin{figure}[H]
    \centering
    \includegraphics[width = \textwidth]{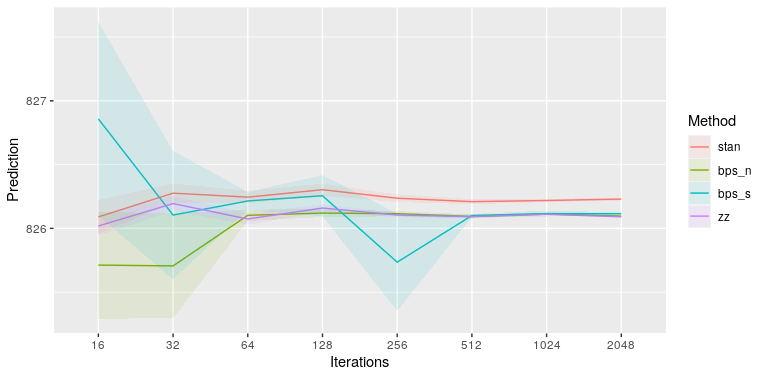}
    \caption{Predictive ability of reversible jump PDMP vs Stan for the robust regression example. The predictive ability is measured by Monte Carlo estimates of the mean square predictive performance.}
\label{fig:predRR}
\end{figure}

\section{Discussion}

We have shown how PDMP samplers can be extended so that they can sample from the joint posterior over model and parameters in variable selection problems. There are a number of open challenges that stem from this work. First, as with any MCMC algorithm, the reversible jump PDMP samplers have tuning parameters. Our simulation results were based on choosing these after empirically evaluating the performance of the samplers on one simple problem. Whilst the samplers mixed well, it is likely that better mixing could be achieved if more informed choices of tuning parameters were made, and theory for guiding such choices is needed. 

Second, the form of our reversible jump PDMP samplers is based on particularly features of the variable selection problem. In other model choice settings, different trans-dimensional moves may be needed. The theory we developed should be able to be adapted to give rules for choosing rates of such moves. Also our trans-dimensional moves are reversible, that is they balance probability flow from model $i$ to model $j$ by the flow of probability from model $j$ to model $i$ -- it would be interesting to see if non-reversible trans-dimensional moves could be constructed.

Third, it is likely that the reversible jump PDMP samplers will still struggle in situations where the posterior is multi-modal with well separated modes. For such cases it would be interesting to try and incorporate ideas such as tempering \cite[]{marinari1992simulated} to allow for better mixing across modes. 
Also given the close links between PDMP samplers and Hamiltonian Monte Carlo, it would be interesting to see if the ideas we have presented for trans-dimensional jumps could be adapted to be used with Hamiltonian Monte Carlo algorithms.

\bibliographystyle{royal}
\bibliography{biblio}

\newpage
\appendix
\setcounter{page}{1}
\begin{center}
    {\bf \Large Supplementary Material: Reversible Jump PDMP Samplers for Variable Selection}
\end{center}

\section{Proofs}
\subsection{Proof of Theorem \ref{th:extended-gen}}
\label{sec:proof-th1}

Recall theorems 34.15 and 34.19 from \cite{PDMPDavis}:
\begin{theorem}[34.15]
\label{th:Davis-3415}
Suppose that $\mu$ is a stationary distribution for $(\bm{x}_t)$. Then there exists a finite measure $\sigma$ on $\Gamma$ (the active boundary) such that for all $t>0$:
\[
    \mathbb{E}_\mu \frac{1}{t} \int_0^t f(\bm{x}_{s^-}) \mbox{d}p^*_s = \int_\Gamma f(\bm{x}) \sigma(\mbox{d}\bm{x})
\]
where $p^*$ is the counting process that counts the number of jumps for the boundary
\end{theorem}
\begin{theorem}[34.19]
\label{th:Davis-3419}
Suppose $\mu$ is a stationary distribution for $(\bm{x}_t)$ and let $\sigma$ be the corresponding boundary measure defined above. Then:
\[
    \int_E \textswab{A} f(\bm{x}) \mu(\mbox{d}\bm{x}) + \int_\Gamma \textswab{C} f(\bm{z}) \sigma(\mbox{d}\bm{z}) = 0
\]
where $\textswab{A}$ is the extended generator and $\textswab{C} f = (Qf-f)$.
\end{theorem}

\begin{remark}
    The construction of PDMPs with an active boundary we presented differs slightly from the one found in \cite[p. 57]{PDMPDavis}. However our process still fits the definition of \cite{PDMPDavis}: one simply needs to separate each $\mathcal{X}^k$ in two at the boundary and consider the two parts as two separated spaces. 
\end{remark}
From the expression of the extended generator given in (26.14) of \cite{PDMPDavis}, it follows that
\[
\sum_{i} \int_{E_i} \mathcal{A}_i f \mbox{d}\mu
    + \sum_{(i,j) \in \mathcal{T}} \int_{E_j} \beta_{i,j}(\bm{z}) \left( \int_{g_{i,j}^{-1}(\{\bm{z}\})} [f(\bm{z}') - f(\bm{z})] Q_{i,j}(\bm{z}'|\bm{z}) \, \mbox{d}\bm{z}'\right) \mu(\bm{z}) \mbox{d}\bm{z}
\]
corresponds to the first term of Theorem \ref{th:Davis-3419}.
It is now left to show that 
\[
    \sum_{(i,j) \in \mathcal{T}} p_{i,j} \int_{\Gamma_{i,j}} \left[ f(g_{i,j}(\bm{z})) - f(\bm{z}) \right] \mu(\bm{z}) |\langle \bm{v},\bm{n}_{i,j}\rangle| \, \mbox{d}\bm{z}
\]
corresponds to the boundary term of  Theorem \ref{th:Davis-3419}. Following \cite{PDMPDavis} notations, for a fixed $(i,j)$ in $\mathcal{T}$, the operator $\textswab{C} = (Qf -f)$ where $Q$ is the jump kernel at the boundary has the following expression:
\[
    \textswab{C}f = p_{i,j}\left[ f(g_{i,j}(\bm{z})) - f(\bm{z}) \right].
\]
Finally, we need to find an expression for measure $\sigma$. We start from the definition in Theorem 4:
\[
    \mathbb{E}_\mu \frac{1}{t} \int_0^t f(\bm{z}_{s^-}) \mbox{d}p^*_s = \int_{\Gamma_{i,j}} f(\bm{z}) \sigma(\mbox{d}\bm{z}),
\]
where $p^*_s$ is the counting process that counts the number of times we hit $\Gamma_{i,j}$.
This holds for all $t$, and we will derive $\sigma$ by considering the value of the left-hand side as $t\rightarrow 0$. Intuitively the idea is that, in this limit, the impact of the events will be negligible and we can evaluate the left-hand side by considering a process where $\bm{z}_0$ is simulated from $\mu$ but then is purely deterministic, with dynamics given by the deterministic part of the PDMP.

As indicator functions of compact sets separate measures on $\Gamma_{i,j}$ we can restrict ourselves to these functions. Let $K$ be a compact of $\Gamma_{i,j}$ and $f$ its characteristic function, $f = 1_K$. We also introduce some additional notation. Let $\tilde{\bm{z}}_t$ be a deterministic process that follows the dynamics of the deterministic part of the PDMP, and let $\tilde{p}_t^*$ be the associated process that counts the number of times $\tilde{\bm{z}}_t$ hits $\Gamma_{i,j}$. For our PDMP let $\Omega_0(t)$ be the set of paths for which no events occur by time $t$, and $\Omega_1(t)$ the set of paths with at least one event prior to time $t$. We can write
\[
\int_0^t 1_K(\bm{z}_{s^-}) \mbox{d}p^*_s = \int_0^t 1_K(\bm{z}_{s^-})1_{\Omega_0(t)} \mbox{d}p^*_s + \int_0^t 1_K(\bm{z}_{s^-})1_{\Omega_1(t)} \mbox{d}p^*_s,
\]
that is split the integral into two, with the first being the contribution from paths with no event and the second the contribution of paths with at least one event.

We now introduce two sets. Let $K_t = \{\bm{z} \in E | \bm{z}_s \in K \text{ with } \bm{z}_0 = z \text{ and } s\in[0,t] \}$, and $\tilde{K}_t= \{z \in E_i | \tilde{\bm{z}}_s \in K \text{ with } \tilde{\bm{z}}_0 = \bm{z} \text{ and } s\in[0,t] \}$. The first is the set of all possible starting points for our PDMP process for which it is possible to hit $K$ by time $t$, the latter is the same set but for the purely deterministic process (equivalently, all PDMP paths with no events by time $t$). As $K$ is compact so is $\tilde{K}_t$. Furthermore as the velocity space is bounded and $K$ is compact so is $K_t$. This, together with our assumption on the event rate means that we can upper bound the event rate for both $\bm{z}\in K_t$ and $\bm{z}\in \tilde{K}_t$, and denote such an upper bound by $\lambda^+$.

Now $\mathbb{E}_\mu \int_0^t 1_K(\bm{z}_{s^-})1_{\Omega_1(t)} \mbox{d}p^*_s$ can be bounded as for the integral to be non-zero we need $\bm{z}_0\in K_t$ and the number of times the PDMP hits $\Gamma_{i,j}$ is at most twice the number of events. Thus an upper bound is $2\lambda^+ t \Pr_\mu(\bm{z}_0\in K_t)$. As the volume of $K_t$ tends to 0 as $t\rightarrow 0$, this upper bound is $o(t)$.

To bound the contribution from the other integral, we notice that we can bound this from above by the the PDMP process without events, and from below by using our bound on the event rate. Thus
\[
\mathbb{E}_\mu \int_0^t 1_K(\tilde{\bm{z}}_{s^-}) \mbox{d}\tilde{p}^*_s \geq
\mathbb{E}_\mu \int_0^t 1_K(\bm{z}_{s^-})1_{\Omega_1(t)} \mbox{d}p^*_s \geq
\exp\{-t\lambda^+\} \mathbb{E}_\mu \int_0^t 1_K(\tilde{\bm{z}}_{s^-}) \mbox{d}\tilde{p}^*_s.
\]

Thus
\[
\mathbb{E}_\mu \int_0^t 1_K(\bm{z}_{s^-})1_{\Omega_1(t)} \mbox{d}p^*_s = \mathbb{E}_\mu \int_0^t 1_K(\tilde{\bm{z}}_{s^-}) \mbox{d}\tilde{p}^*_s +o(t).
\]
Thus, for any compact subset of $\Gamma_{i,j}$, $K$, we have
\[
\int_{\Gamma_{i,j}} 1_K(\bm{z}) \sigma(\mbox{d}\bm{z}) = \lim_{t\rightarrow 0} \mathbb{E}_\mu \frac{1}{t} \int_0^t 1_K(\bm{z}_{s^-}) \mbox{d}p^*_s =
\lim_{t\rightarrow 0} \mathbb{E}_\mu \int_0^t 1_K(\tilde{\bm{z}}_{s^-}) \mbox{d}\tilde{p}^*_s.
\]

The final step is to obtain $\lim_{t\rightarrow 0} \mathbb{E}_\mu \int_0^t 1_K(\tilde{\bm{z}}_{s^-}) \mbox{d}\tilde{p}^*_s$, which comes from the following result for solutions of ordinary differential equations: 



\begin{theorem}
\label{th:deterministic-flux}
Let $E$ be a subset of $\mathbb{R}^d$, $\mu$ a measure with a density with respect to the Lebesgue measure of $E$ and let $X$ be the vector field associated to the order ordinary differential equation
\[
    \frac{\mbox{d}(\bm{z}_t)}{\mbox{d}t} = X(\bm{z}_t).
\]
Let $H$ be an hyperplane of $E$ with normal $n_H$ and $t_*(\bm{z})$ be the first hitting time of a trajectory with $H$ starting from $z$. Assuming there exists a measure $\sigma$ on $H$ such that for all $f:H \rightarrow \mathbb{R}$:
\[
  \lim_{t\rightarrow 0} \frac{1}{t} \int_{E} f(\bm{z}_{t_*(\bm{z}_0)}) 1_{t_*(\bm{z}_0) \leq t} \, \mu(\bm{z}_0) \, \mbox{d}\bm{z}_0 = \int_H f(z) \sigma(\bm{z}) \mbox{d}\bm{z}
\]
then $\sigma$ has a density with respect to the Lebesgue measure of $H$ and 
\[
    \sigma(\mbox{d}\bm{z}) = \mu(\bm{z}) |\langle \bm{z}, \bm{n}_H \rangle| \mbox{d}\bm{z}
\]
\end{theorem}
For completeness, this result is proven in Appendix \ref{sec:proof_flux}. The intuition behind the $|\langle \bm{z}, \bm{n}_H \rangle|$ factor is that `area' of space that hits $\bm{z}$ within an infinitesimal time interval $\mbox{d}t$ is $|\langle \bm{z}, \bm{n}_H \rangle|\mbox{d}t$.

Applying this theorem to the deterministic part of our PDMP we get
\[
    \sigma(\bm{z}) \mbox{d}\bm{z}= \mu(\bm{z}) |\langle \bm{v},\bm{n}_{i,j}\rangle| \mbox{d}\bm{z}
\]
where $\bm{n}_{i,j}$ is the normal the the boundary $\Gamma_{i,j}$, which concludes the proof.

\subsection{Proof of Theorem \ref{th:invariant}}
\label{sec:proof-th2}
\begin{proof}
The core of the proof is rewriting (\ref{eq:th1}) of Theorem \ref{th:extended-gen} in a more suitable way.
First, notice that 
\[
\int_{\Gamma_{i,j}} f(g_{i,j}(\bm{z})) \mu(\bm{z}) |\langle \bm{v},\bm{n}_{i,j}\rangle| \, \mbox{d}\bm{z} = \int_{E_j} f(\bm{z}) \bar{\nu}_{i,j}(\bm{z}) \, \mbox{d}\bm{z}
\]
by definition of the pushforward measure $\bar{\nu}_{i,j}$.
Second, since $Q_{i,j}$ integrates to 1, we get:
\[
\int_{E_j} \beta_{i,j}(\bm{z}) \int_{g_{i,j}^{-1}(\{\bm{z}\})} f(\bm{z}) Q_{i,j}(\bm{z}'|\bm{z}) \, \mbox{d}\bm{z}' \mu(\bm{z}) \, \mbox{d}\bm{z} = \int_{E_j} \beta_{i,j}(\bm{z}) f(\bm{z}) \mu(\bm{z}) \, \mbox{d}\bm{z}.
\]
Third, using the definitions of $G_{i,j}$
\[
    \int_{g_{i,j}^{-1}(\{\bm{z}\})} f(\bm{z}') Q_{i,j}(\bm{z}'|\bm{z}) \, \mbox{d}\bm{z}' = \int_U f(G_{i,j}(\alpha,\bm{z})) Q_{i,j}(\alpha|\bm{z}) \, \mbox{d}\alpha
\]
Finally, using a change the change of variable from $\bm{z}'=(\bm{x}',\bm{v}')$ to $\alpha,\bm{z}$ defined as $\bm{z}'=G_{i,j}(\alpha,\bm{z})$, and the definition of $\nu_{i,j}$:
\[
     \int_{\Gamma_{i,j}} f(\bm{z}') \nu(\bm{z}') |\langle \bm{v}',\bm{n}_{i,j}\rangle| \, \mbox{d}\bm{z}' = \int_{E_j} \int_U f(G_{i,j}(\alpha,\bm{z})) \nu_{i,j}(G(\alpha,\bm{z})) |J_{G_{i,j}}(\alpha,\bm{z})| \, \mbox{d}\alpha \, \mbox{d}\bm{z}.
\]
Thus (\ref{eq:th1}) of Theorem \ref{th:extended-gen}, evaluated for $\pi=\nu$, can be restated as:
\begin{align}
\begin{split}
    0
    &= \sum_{i} \int_{E_i} \mathcal{A}_i f \mbox{d}\nu+ \sum_{(i,j) \in \mathcal{T}} \int_{E_j}  \left[ p_{i,j} \bar{\nu}_{i,j}(\bm{z})  - \beta_{i,j}(\bm{z}) \nu(\bm{z}) \right] f(\bm{z})\mbox{d}\bm{z}   \\ 
    &+ \sum_{(i,j) \in \mathcal{T}} \int_{E_j} \int_U \left[ \beta_{i,j}(\bm{z}) Q_{i,j}(\alpha|\bm{z}) \nu(\bm{z}) - p_{i,j} \nu_{i,j}(G_{i,j}(\alpha,\bm{z})) |J_{G_{i,j}}(\alpha,\bm{z})| \right] f(G_{i,j}(\alpha,\bm{z})) \, \mbox{d}\alpha \, \mbox{d}\bm{z}
    \end{split}
\end{align}
By construction, for any function $f\in \mathcal{F}$:
\[
    \sum_i \int_{E_i} \mathcal{A}_i f \mbox{d}\nu = 0.
\]
Conditions \ref{eq:thm2.1} and \ref{eq:thm2.2} of the theorem clearly imply that
\[
\sum_{(i,j) \in \mathcal{T}} \int_{E_j}  \left[ p_{i,j} \bar{\nu}_{i,j}(\bm{z})  - \beta_{i,j}(\bm{z}) \nu(\bm{z}) \right] f(\bm{z})\mbox{d}\bm{z} = 0
\]
and
\[
\sum_{(i,j) \in \mathcal{T}} \int_{E_j} \int_U \left[ \beta_{i,j}(\bm{z}) Q_{i,j}(\alpha|\bm{z}) \nu(\bm{z}) - p_{i,j} \nu_{i,j}(G_{i,j}(\alpha,\bm{z})) |J_{G_{i,j}}(\alpha,\bm{z})| \right] f(G_{i,j}(\alpha,\bm{z})) \, \mbox{d}\alpha \, \mbox{d}\bm{z} = 0
\]
respectively, which concludes the proof.
\end{proof}

\subsection{Proof of Theorem \ref{th:invariantDiscrete}}
\label{sec:proof-th3}
\begin{proof}
The proof is largely the same as the proof of Theorem \ref{th:invariant}. The only change is how to treat the change of variable $\bm{z}' = G_{i,j}(\alpha,\bm{z})$. First, notice that $G_{i,j}$ leaves the position invariant since $g_{i,j}$ leaves the position invariant. Then, since the velocity space is discrete and $\alpha \mapsto G_{i,j}(\alpha,\bm{z})$ is a one to one mapping for a fixed $\bm{z} \in E_{j}$:
\begin{align*}
    \int_{\Gamma_{i,j}} f(\bm{z}') \nu_{i,j}(\bm{z}') |\langle \bm{v}',\bm{n}_{i,j}\rangle| \, \mbox{d}\bm{z}' &= \int_{\mathcal{X}^j} \sum_{\bm{v}' \in \mathcal{V}^i} f(\bm{x}',\bm{v}') \nu_{i,j}(\bm{x}',\bm{v}')  \, \mbox{d}\bm{x}' \\
    &= \int_{\mathcal{X}^j} \sum_{\bm{v}\in E_j} \sum_{\alpha \in U} f(G_{i,j}(\alpha,(\bm{x},\bm{v})) \nu_{i,j}(G_{i,j}(\alpha,(\bm{x},\bm{v}))) \, \mbox{d}\bm{x}.
\end{align*}
which we can write in integral form as
\[
\int_{\Gamma_{i,j}} f(\bm{z}') \nu_{i,j}(\bm{z}') |\langle \bm{v}',\bm{n}_{i,j}\rangle| \, \mbox{d}\bm{z}' =
\int_{E_j} \int_{U} f(G_{i,j}(\alpha,(\bm{x},\bm{v})) \nu_{i,j}(G_{i,j}(\alpha,(\bm{x},\bm{v}))) \, \mbox{d}\alpha \mbox{d}\bm{z}.
\]
\end{proof}

\subsection{Proof of Proposition \ref{prop:1}}
\begin{proof} 
Here:
\[
    \mathcal{V}^i = \{\bm{v} \in \mathbb{R}^{|\gamma^i|} \text{ such that } \|\bm{v}\| = 1\}
\]
For $\bm{z} \in E_i$:
\[
    \nu(\bm{z}) = \pi_i(\bm{x}) \frac{1}{A_{sphere}(|\gamma^i|)} 
\]
with $A_{sphere}(|\gamma^i|) = \frac{\Gamma(\frac{|\gamma^i|}{2})}{2\pi^{\frac{|\gamma^i|}{2}}}$ the area of the unit sphere of $\mathbb{R}^{|\gamma^i|}$.

In this case we can define $G_{i,j}: (-1,1) \times E_j\rightarrow E_i$ such that 
$\bm{z}'=(\bm{x}',\bm{v}')=G_{i,j}(\alpha,\bm{z})$, where $\bm{z}=(\bm{x},\bm{v})$, as $\bm{x}'=\bm{x}$ and
\[
     \bm{v}'=\sqrt{1 - \alpha^2} \bm{v} + \alpha \bm{n}_{i,j}.
\]

Let $B \subset \mathcal{V}^j$
\begin{align*}
    \int_B \bar{\nu}_{i,j}(\bm{x},\bm{v}) \mbox{d}\bm{v} &= \pi_i(\bm{x})\int_{g_{i,j}^{-1}(B)} |<\bm{v}',\bm{n}_{i,j}>| \frac{1}{A_{sphere}(|\gamma^i|)} \mbox{d}\bm{v}' \\
    &= \pi_i(\bm{x}) \frac{1}{A_{sphere}(|\gamma^i|)} \int_{[-1,1] \times B} |< \sqrt{1 - \alpha^2} \bm{v} + \alpha \bm{n}_{i,j} ,\bm{n}_{i,j}>| |J_{G_{i,j}}| \mbox{d}\alpha \mbox{d}\bm{v} \\
    &= 2 \pi_i(\bm{x})  \frac{1}{A_{sphere}(|\gamma^i|)} \int_{(0,1] \times B} \alpha |J_{G_{i,j}}| \mbox{d}\alpha \mbox{d}\bm{v},
\end{align*}
where the last inequality uses the symmetry of the integral with respect to $\alpha$, and that $<\bm{v},\bm{n}_{i,j}>=0$ by definition of velocities in $E_j$.

The determinant of $G_{i,j}$ must be carefully computed since the velocities lives in the sphere and not on the full vector space. We get:
\[
    |J_{G_{i,j}}(\alpha,\bm{v})| = \sqrt{1-\alpha^2}^{|\gamma^j|-2}.
\]
Therefore for $|\gamma^j| > 0$:
\begin{align*}
    \bar{\nu}_{i,j}(\bm{z}) &= \pi_i(\bm{x}) \frac{2}{A_{sphere}(|\gamma^i|)}.     \frac{1}{|\gamma^j|}
\end{align*}
Hence for $|\gamma^j| > 0$:
\begin{align}
    \beta_{i,j} &= p_{i,j} \frac{\pi_i(\bm{x})}{\pi_j(\bm{x})}  \frac{2 A_{sphere}(|\gamma^j|)}{A_{sphere}(|\gamma^i|)} \frac{1}{|\gamma^j|},
\end{align}
and
\begin{align}
    Q_{i,j}(\alpha|\bm{z}) &= \frac{|\alpha| |\gamma^j| \sqrt{1-\alpha^2}^{|\gamma^j|-2}}{2} \text { for } \bm{z} \in E_j \text{ and } \alpha \in (-1,1)
\end{align}
For $|\gamma^j| = 0$, BPS and ZigZag are equivalent, thus we use ZigZag rates.
\end{proof}

\subsection{Proof of Proposition \ref{prop:2}}
\begin{proof} 
Here:
\[
    \mathcal{V}^i = \mathbb{R}^{|\gamma^i|}
\]
Therefore, for $\bm{z} \in E_i$:
\[
    \nu(\bm{z}) = \pi_i(\bm{x}) \frac{1}{(2\pi)^{\frac{|\gamma^i|}{2}}} e^{-\frac{1}{2}\|\bm{v}\|^2}
\]
For $(i,j) \in \mathcal{T}$ with $|\gamma^j| > 0$, we can define $G_{i,j}: (-1,1) \times E_j\rightarrow E_i$ such that 
$\bm{z}'=(\bm{x}',\bm{v}')=G_{i,j}(\alpha,\bm{z})$, where $\bm{z}=(\bm{x},\bm{v})$, as $\bm{x}'=\bm{x}$ and
\[
    \bm{v}' = \bm{v} + \alpha \bm{n}_{i,j}.
\]
Clearly, we have $|J_{G_{i,j}}(\alpha,\bm{z})| = 1$.
Let $B \subset \mathcal{V}^j$
\begin{align*}
    \int_B \bar{\nu}_{i,j}(\bm{x},\bm{v}) \mbox{d}\bm{v} &= \pi_i(\bm{x})\int_{g_{i,j}^{-1}(B)} |<\bm{v}',\bm{n}_{i,j}>| \frac{1}{(2\pi)^{\frac{|\gamma^i|}{2}}} e^{-\frac{1}{2}\|\bm{v}\|^2} \mbox{d}\bm{v}' \\
    &= \pi_i(\bm{x}) \frac{1}{(2\pi)^{\frac{|\gamma^i|}{2}}} \int_{\mathbb{R} \times B} |<\bm{v}+\alpha \bm{n}_{i,j},\bm{n}_{i,j}>| e^{-\frac{1}{2}\| \bm{v}+\alpha \bm{n}_{i,j} \|^2} |J_{G_{i,j}}| \mbox{d}\alpha \mbox{d}\bm{v} \\
    &= \pi_i(\bm{x}) \frac{1}{(2\pi)^{\frac{|\gamma^i|}{2}}} \int_{\mathbb{R} \times B} |\alpha| e^{-\frac{1}{2}(\alpha^2 + \|\bm{v}\|^2)} \mbox{d}\alpha \mbox{d}\bm{v}\\
    &= \pi_i(\bm{x}) 2 \frac{1}{(2\pi)^{\frac{|\gamma^i|}{2}}} \int_{B} e^{-\frac{1}{2}(\|\bm{v}\|^2)} \mbox{d}\bm{v}.
\end{align*}
Therefore for $|\gamma^j| > 0$:
\begin{align*}
    \bar{\nu}_{i,j}(\bm{z}) &= \pi_i(\bm{x}) 2 \frac{1}{(2\pi)^{\frac{|\gamma^i|}{2}}} e^{-\frac{1}{2}(\|\bm{v}\|^2)}, 
\end{align*}
and
\begin{align}
    \beta_{i,j} &= p_{i,j} \frac{\pi_i(\bm{x})}{\pi_j(\bm{x})}  \frac{2}{\sqrt{2\pi}}
\end{align}
Furthermore, $|J_{G_{i,j}}| = 1$ thus:
\begin{align}
    Q_{i,j}(\alpha|\bm{z}) &= 2 |\alpha| e^{-\frac{1}{2}\alpha^2} \text { for } \bm{z} \in E_j \text{ and } \alpha \in \mathbb{R}
\end{align}
\end{proof}

\subsection{Proof of Theorem \ref{th:deterministic-flux}} \label{sec:proof_flux}

\begin{proof}
Let  $Z(t,\bm{z})$ be the flow associated to the ODE. Fix $\bm{z}_0$ in $H$ such that $a=\langle X(\bm{z}_0),\bm{n}_H\rangle\neq 0$.  In order to show that $\sigma$ has the required form it is enough to prove that
\[
\lim_{r\rightarrow 0}\frac{\sigma(H\cap B(\bm{z}_0,r))}{\mbox{Volume}_H(H\cap B(\bm{z}_0,r))}=|a|\mu(\mbox{z}_0)
\]
where $B(\bm{z}_0,r)$ is the ball of radius $r$ centered on $\bm{z}_0$, $H\cap B(\bm{z}_0,r)$ is the slice of this ball that lies on $H$, and $\mbox{Volume}_H$ is the Lebesgue measure of the hyperplane  $H$. 

For $r$ and $t>0$, let 
\[
E_{r,t}=\{\bm{z}\in E: 0\leq t_*(\bm{z})\leq t,\, Z(t_*(\bm{z}),\bm{z})\in B(\bm{z}_0,r) \cap H\}
\]
We can re-write the definition of $\sigma$ in terms of $E_{r,t}$, from which we are reduced to show that
\[
\lim_{r\rightarrow 0}\lim_{t\rightarrow 0}\frac{\tfrac1t \mu(E_{r,t})}{\mbox{Volume}_H(H\cap B(\bm{z}_0,r))}=|a|\mu(\bm{z}_0)
\]
The idea of the proof is that to calculate $\mu(E_{r,t})$ we need to find the volume of $E_{r,t}$ as $r,t\rightarrow 0$, and we will do this by bounding $E_{r,t}$ from both above and below by cylinders. We can approximated these cylinders by the set $E_{r,t}$ that we would obtain if $X(\bm{z})=X(\bm{z}_0)$ (i.e. the derivative was constant). The following lemma enables us to quantify the order of error for such an approximation.

\begin{lemma}
\begin{enumerate}
    \item Let $r_0 > 0$. When $t>0$ is small enough, $Z(s,\bm{z})\in B(\bm{z}_0,r_0)$ for all $z\in B(\bm{z}_0,r_0/2)$ and $s\in [-t,t]$.
    \item Let $e(s,\bm{z})$ be the function defined by $e(s,\bm{z})=Z(s,\bm{z})-(\bm{z}+s X(\bm{z}_0))$. Then:
    \[
        \lim_{(s,\bm{z})\rightarrow (0,\bm{z}_0), s\neq 0}\tfrac1s e(s,\bm{z})=0.
    \]
\end{enumerate}
\label{lemma:error}
\end{lemma}
\begin{proof}
For part (1): $Z$ is continuous as a function from $\mathbb{R}\times E$ to $E$. Therefore for $t$ small enough, $Z(s,z)\in B(\bm{z}_0,r_0)$ for all $\bm{z}\in B(\bm{z}_0,r_0/2)$ and $s\in [-t,t]$.

For part (2): We define 
\[
    C_{\bm{z}_0,r} = \sup_{\bm{z} \in B(\bm{z}_0,r)} \|X(Z(s,\bm{z})) - X(\bm{z}_0)\|.
\]
We clearly have $C_{\bm{z}_0,r} < \infty$ and $C_{\bm{z}_0,r} \xrightarrow[r\rightarrow 0]{} 0$. We set $r_0 = 2 \|\bm{z} - \bm{z}_0\|$, and
using the mean value theorem with the function $e_{\bm{z}}(s)=e(s,\bm{z})$ whose derivative is $X(Z(s,\bm{z}))-X(\bm{z}_0)$ with the fact that $e(0,\bm{z})=0$ implies 
\[
    \|e(s,\bm{z})\| \leq s C_{\bm{z}_0,2\|\bm{z}-\bm{z}_0\|}
\]
which concludes the proof.
\end{proof}

The following two lemmas then give, respectively, the cylinders that lower and upper bound $E_{r,t}$.
\begin{lemma} For $z\in B(\bm{z}_0,r) \cap H$ and $s \in [-t,0]$, 
there exists a constant $C>0$ and function $o(r)\geq 0$ with $lim_{r\rightarrow 0} o(r)/r=0$ such that
\[
E_{r,t}\subset \left(B(\bm{z}_0,r+tC)\cap H\right) + [-at-o(r)t,0]\bm{n}_H
\]
\label{lemma:bound1}
\end{lemma}
The addition on right-hand side of this statement should be interpreted as addition of sets on a product space. That is the set on the right-hand side is the cylinder of points that can be written as $\bm{z}+s \bm{n}_H$ where $\bm{z} \in B(\bm{z}_0,r+tC)\cap H$  and $s\in [-at-o(r)t,0]$.

\begin{proof}
We have that
\[
E_{r,t}=\{Z(s,\bm{z}):-t\leq s\leq 0,\, \bm{z}\in H\cap B(\bm{z}_0,r)\}.
\]
Let $\bm{z}$ be in $B(\bm{z}_0,r) \cap H$ and $s \in [-t,0]$
\begin{align*}
Z(s,\bm{z})&=\bm{z}+sX(\bm{z}_0)+e(s,\bm{z})\\
&=\bm{z}+s(X(\bm{z}_0)-a\bm{n}_H)+sa\bm{n}_H+e(s,\bm{z})
\end{align*}
Notice that $X(\bm{z}_0) - a\bm{n}_H$ is in $H$ hence $\bm{z} + s(X(\bm{z}_0) - a \bm{n}_H)$ is also in $H$.
Lemma \ref{lemma:error} allows us to bound $\frac{1}{s}e(s,\bm{z})$ on $[0,t]\times B(\bm{z}_0,r)\cap H$  which concludes the proof.
\end{proof}

\begin{lemma}
For $\bm{z}\in B(\bm{z}_0,r) \cap H$ and $s \in [-t,0]$, there exists a constant $C>0$ and function $o(r)\geq 0$ with $lim_{r\rightarrow 0} o(r)/r=0$ such that
\[
 (B(\bm{z}_0,r-tC)\cap H)+[-at+o(r)t,0]\bm{n}_H\subset E_{r,t}.
\]
\label{lemma:bound2}
\end{lemma}
\begin{proof}
Let $\bm{z} \in B(\bm{z}_0,r-tC)\cap H-[0,at-o(r)t]n_H$.
\begin{align*}
\langle Z(s,\bm{z}), \bm{n}_H \rangle &= \langle \bm{z}, \bm{n}_H \rangle + s a + \langle e(s,\bm{z}), \bm{n}_H \rangle
\end{align*}
Using Lemma \ref{lemma:error}, there exists $s\in [0,at - o(r)]$ (intuitively $s \approx \langle \bm{z}_0 - \bm{z},\bm{n}_H\rangle / a \leq \frac{at}{a} = t$) such that $\langle Z(s,\bm{z}), \bm{n}_H \rangle = \langle \bm{z}_0,\bm{n}_H\rangle$, which implies $Z(s,\bm{z})\in H$. Hence $\bm{z}\in E_{r,t}$.
\end{proof}

Using Lemmas \ref{lemma:bound1} and \ref{lemma:bound2}, we have constants $C_1$ and $C_2$ and a function $o_1(r)$ and $o_2(r)$ with $o_k(r)/r\rightarrow 0$ ($k=1,2$) such that
\begin{align*}
    &\frac{\mbox{Volume}(E_{r,t})}{t} \geq (a -o_1(r)) \mbox{Volume}_H(B(\bm{z}_0,r-t C_1)\cap H) \\
    &\frac{\mbox{Volume}(E_{r,t})}{t} \leq (a + o_2(r)) \mbox{Volume}_H(B(\bm{z}_0,r+tC_2)\cap H)
\end{align*}
where $\mbox{Volume}$ and $\mbox{Volume}_H$ are with respect to the Lebesgue measure of $E$ and $H$ respectively. Since $\mu$ has continuous densities, for any $\epsilon > 0$ there exists $r_0$ such that for every $r < r_0$, 
\[
    \left|\frac{\mu(E_{r,t})}{t} - \frac{\mu(\bm{z}_0) \mbox{Volume}(E_{r,t})}{t}\right| < \epsilon.
\]
Hence for every $\epsilon > 0$ there exists $r_0 > 0$ and $T > 0$ such that for every $r < r_0$ and $0 < t < T$,
\[
\left|\frac{\tfrac1t \mu(E_{r,t})}{\mbox{Volume}_H(H\cap B(\bm{z}_0,r))} - |a|\mu(\bm{z}_0) \right| < \epsilon
\]
which concludes the proof.

\end{proof}
\section{Sensitivity to tuning parameters}
The tuning parameters for the reversible jump PDMP methods consist of:
\begin{enumerate}
\item the reversible jump parameter $p_{i,j}$ introduced in Section 3;
\item the refreshment distribution for velocities of BPS; and
\item the refreshment rate $\lambda^{ref}$ for BPS.
\end{enumerate}
To investigate the performance of our methods we consider sampling from a 100-dimensional spike and slab distribution with $\theta_j \sim~ 0.5\mathcal{N}(0.5, 1) + 0.5\delta_0,$ independently for $j=1,...,100$. Each sampler was run 100 times for a computational budget of 10. 

\begin{figure}[H]
    \centering
    \includegraphics[width = .7\textwidth]{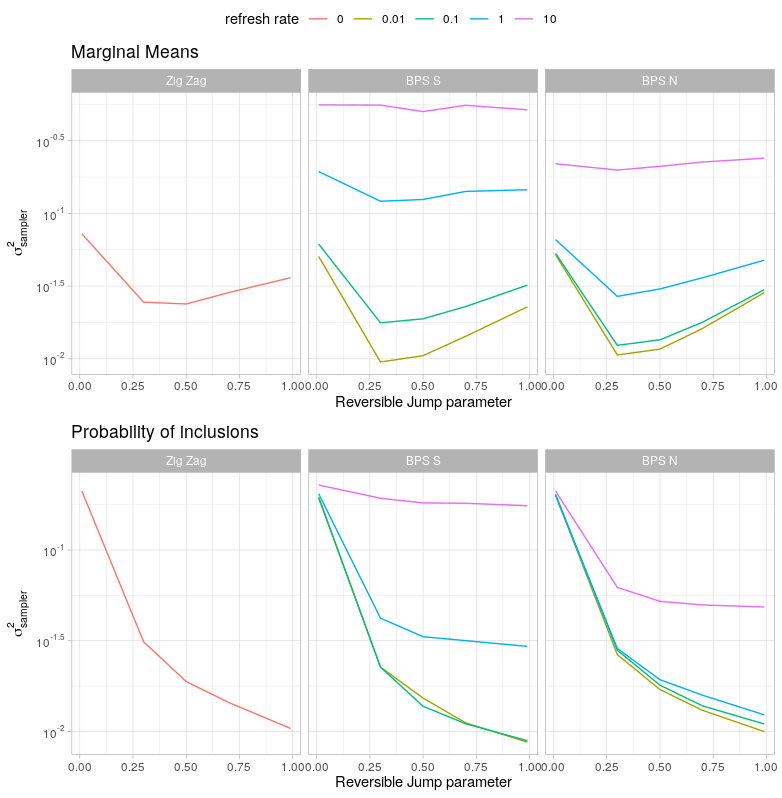}
    \caption{Plots of the Monte Carlo variance for different samplers with varying choices of tuning parameters. The Monte Carlo variance is defined relative to the estimates of the posterior means and probabilities of inclusion. }
    \label{fig:GaussMeans}
\end{figure}

We compare the methods based on the statistical efficiency of the estimates of marginal means and probability of inclusions.  The PDMP methods appear to perform optimally in estimation of the marginal mean when they are able to spend a reasonable amount of time exploring the model space and exploring the parameter space. This can be seen in Figure \ref{fig:GaussMeans} where optimal statistical efficiency for the marginal means occurs when $p_{i,j}$ is the range 0.2 to 0.7 across all samplers. Figure \ref{fig:GaussMeans} also shows that higher values of the reversible jump parameter $p_{i,j}$ gives better estimation of the probabilities of inclusions. Having a high refreshment rate appears to impact BPS with Normal velocity distribution less than BPS with velocities distributed uniformly on the hyper-sphere. This is seen as the magenta curve (refreshment = 10) is lower for BPS N than BPS S in Figure \ref{fig:GaussMeans} across all reversible jump parameters. All BPS methods appear to out-perform Zig-Zag in terms of marginal mean estimation when the refreshment rate is low. 

\section{General implementation details}
Inference in Bayesian model selection relies on expectations with respect to a posterior target distribution, $\pi(\bm{\theta}, \bm{\gamma})$. The parameters are $\bm{\theta}$ while $\bm{\gamma}$ is a vector which indexes the model with elements $\gamma_j = 1$ if the $j$th variable is included and $\gamma_j = 0$ otherwise. The posterior has the form
\begin{align*}
\pi(\bm{\theta}, \bm{\gamma}) \propto L(y^{1:n}|\bm{\theta},\bm{\gamma})\pi_0(\bm{\theta} |\bm{\gamma})\pi_0(\bm{\gamma})
\end{align*}
where $L(y^{1:n}|\bm{\theta}, \bm{\gamma})$ defines a likelihood function for observations $y^{1:n}$, $\pi_0(\bm{\theta} |\bm{\gamma})$ and $\pi_0(\bm{\gamma})$ denote prior distribution for $\bm{\theta}$ and $\bm{\gamma}$. We abuse notation writing $\bm{\theta}_{\bm{\gamma}}$ to denote the sub-vector of $\bm{\theta}$ with only the elements where $\gamma_j=1$. Moreover, we write $\pi(\bm{\theta}_{\bm{\gamma}})$ for $\pi(\bm{\theta} \mid \bm{\gamma})$ where $\pi(\bm{\theta} \mid \bm{\gamma}) = 0$ whenever $|\theta_j|>0$ with corresponding $\gamma_j = 0$. 
\\
When simulating from the reversible jump PDMP sampler there are two types of events: normal events for the PDMP sampler within a model $\bm{\gamma}$ and model jump events. The standard PDMP events are taken with respect to $\pi(\bm{\theta}|\bm{\gamma})$ so rates to sample are given using the usual Bouncy Particle Sampler or Zig-Zag rates on the conditioned model
\begin{align*}
\lambda^{BPS}(s) &= \left(-\bm{v}_{\bm{\gamma}} \cdot \nabla_{\bm{\theta}_{\bm{\gamma}}}\log\pi(\bm{\theta}_{\bm{\gamma}} + sv_{\bm{\gamma}}) \right)^+, \\ 
\lambda^{ZZ}_i(s) &= \left(-v_{i} \nabla_{\theta_{i}}\log\pi(\bm{\theta}_{\bm{\gamma}} + sv_{\bm{\gamma}}) \right)^+, \qquad \text{for }i\in\{i : \gamma_i = 1\}.
\end{align*}
In practice to simulate these events we first bound the rates by a simple function, which for the examples we consider will be linear in time. We then simulate events from a Poisson process with this linear-in-time rate, which can be done exactly, and use thinning to generate the actual events in the PDMP. Derivations of the linear-in-time bounds on the rates that we use are now given, before we give the rates for jumps between models.

\subsection{Rates for logistic and robust regression}
Here we give details on simulating rates for the logistic regression example. Taking a simple independent Gaussian prior the reintroduction rate simplifies to
$$
\beta_{\bm{\gamma}\rightarrow \bm{\gamma}'} = p_{\bm{\gamma}\rightarrow \bm{\gamma}'} \frac{1}{\sqrt{2\pi\sigma^2}}\frac{w}{(1-w)}
$$
where $\bm{\gamma}$ and $\bm{\gamma}'$ are defined as above. The standard PDMP rates are used for $\pi(\bm{\theta}|\bm{\gamma}) = \pi(\bm{\theta}_{\bm{\gamma}})$. So to simplify notation we will assume a fixed dimension and write $\bm{\theta}$ dropping the indexing with ${\bm{\gamma}}$. The log posterior for both logistic and robust regression can be written in the form
$$
-\log\pi(\bm{\theta}) = \sum_{i=1}^n g(e_i) + \frac{\bm{\theta}^T\bm{\theta}}{2\sigma^2}.
$$
For logistic regression $e_i = -\bm{x}_{i}^T\bm{\theta}$ and $g(e_i) = -\log\left(\frac{\exp(y_i)}{1 + \exp(e_i)}\right)$, and for robust regression $e_i = y_i - \bm{x}_{i}^T\bm{\theta}$ and $g(e_i) = -\log\left( \exp(-\frac{1}{2}e_i^2)+ \frac{1}{10}\exp(-\frac{1}{200}e_i^2)\right)$. We consider bounding the event rates for the Bouncy Particle Sampler and ZigZag below.
\\
\\
\textbf{Bouncy Particle Sampler:} Let $f(\bm{\theta} + t\bm{v}) = -\nabla_{\bm{\theta}}\log \pi(\bm{\theta} + t\bm{v})$ the event rate depends on the quantity 
$$
\langle \bm{v}, f(\bm{\theta} + t\bm{v})\rangle = \langle v, \sum_{i=1}^n\bm{x}_i^Tg'(e_i(t)) + \frac{1}{\sigma^2}(\bm{\theta} + t\bm{v})\rangle.
$$
where $g'(e_i(t))$ is the derivative of $g$ evaluated evaluated at $e_i(t) = -\bm{x}_{i}^T(\bm{\theta}+tv)$ for logistic regression and $e_i(t) = y_i-\bm{x}_{i}^T(\bm{\theta}+tv)$ for robust regression. The in-time derivative of this quantity is
$$
\frac{d}{dt}\langle \bm{v}, f(\bm{\theta} + t\bm{v})\rangle = \langle \bm{v}, \sum_{i=1}^n\bm{x}_i^Tg''(e_i(t))\bm{x}_i^T\bm{v} + \frac{1}{\sigma^2}\bm{v}\rangle = \sum_{i=1}^ng''(e_i(t))(\bm{x}_i^T\bm{v})^2+\frac{1}{\sigma^2}\bm{v}^T\bm{v}.
$$
If the in-time derivative can be bounded by a constant we can simulate using linear rates. For logistic regression $g''(e_i) \leq \frac{1}{4}$ and for robust regression $g''(e_i) < 1$. The Bouncy Particle Sampler rate is bounded by the linear rate
$$
\max(0, \langle \bm{v}, f(\bm{\theta} + t\bm{v})\rangle) \leq \max\left(0, \langle \bm{v}, f(\bm{\theta})\rangle\right) + t\left(\frac{1}{\sigma^2}\bm{v}^T\bm{v}+ c\sum_{i=1}^n(\bm{x}_i^T\bm{v})^2\right)
$$
where $c$ is chosen according to the application. Inversion methods for thinning a Poisson process can be used to simulate the events \cite[]{bierkens2018piecewise}. 
\\
\\
\textbf{ZigZag:} Let $f(\bm{\theta} + t\bm{v}) = -\frac{\rm{d}}{\rm{d}\theta_i}\log \pi(\bm{\theta} + t\bm{v})$ the event rate depends on the quantity 
$$
v_i f(\bm{\theta} + t\bm{v})\rangle = v_i\sum_{i=1}^nx_{ij}g'(e_i(t)) + v_i\frac{1}{\sigma^2}(\theta_i + tv_i)
$$
where $g'(e_i(t))$ is defined as in the BPS rate. The in-time derivative of this quantity is
$$
\frac{d}{dt}v_i f(\bm{\theta} + t\bm{v}) = v_i \sum_{i=1}^nx_{ij}g''(e_i(t))\bm{x}_i^T\bm{v} + v_i^2\frac{1}{\sigma^2} = \sum_{i=1}^ng''(e_i(t))x_{ij}v_i\bm{x}_i^T\bm{v}+ v_i^2\frac{1}{\sigma^2}.
$$
Using the same method as before the ZigZag rate is bounded by the linear rate
$$
\max(0, v_i f(\bm{\theta} + t\bm{v})) \leq \max(0, v_i f(\bm{\theta})) + t\left(v_i^2\frac{1}{\sigma^2}+\sum_{i=1}^n|cx_{ij}v_i\bm{x}_i^T\bm{v}|\right)
$$
where $c$ is chosen according to the application.

\subsection{Rates of jumps between models}

Model jump events occur when a parameter, $\theta_i$, hits a hyper-plane $\theta_i = 0$ and with probability $p_{\bm{\gamma}\rightarrow \bm{\gamma}'}$ we jump to model $\bm{\gamma}'$ where $\gamma_i' = 0$. The other type of model jump event occurs when a variable is reintroduced. For each of the deactivated variables ($\gamma_i = 0$), we simulate a time to reintroduce the variable (switching to $\bm{\gamma}'$ where $\bm{\gamma}'_{-i} = \bm{\gamma}_{-i}$ with $\gamma'_i = 1$). The rate to reintroduce the variable in the Bayesian inference problem is
\begin{align*}
\beta_{\bm{\gamma}' \rightarrow \bm{\gamma}} = p_{\bm{\gamma}' \rightarrow \bm{\gamma}}\frac{L(y^{1:n}|\bm{\theta}_{\bm{\gamma}'})\pi_0(\bm{\theta}_{\bm{\gamma}'})\pi_0(\bm{\gamma}')}{L(y^{1:n}|\bm{\theta}_{\bm{\gamma}})\pi_0(\bm{\theta}_{\bm{\gamma}})\pi_0(\bm{\gamma})},
\end{align*}
where often computational savings are possible since the reintroduced variable will be zero $\theta_i = 0$ and it is often the case that $L(y^{1:n}|\bm{\theta}_{\bm{\gamma}'}) = L(y^{1:n}|\bm{\theta}_{\bm{\gamma}})$. In these cases the rate to reintroduce a variable will only depend on the choice of prior. 

Both examples we consider had a Gaussian spike and slab prior, of the form
\begin{align*}
\bm{\theta}_{\bm{\gamma}} &\sim \mathcal{N}(\bm{\mu}_{\bm{\gamma}}, \bm{\Sigma}_{\bm{\gamma}})\\
\bm{\gamma}&\sim w^{\sum_{j=1}^p\gamma_j}(1-w)^{p - \sum_{j=1}^p\gamma_j},
\end{align*}
for a fixed $w$.  The rate to reintroduce the $i$th variable, jumping from model $\bm{\gamma}$ to $\bm{\gamma}'$ where $\bm{\gamma}'_{-i} = \bm{\gamma}_{-i}$ with $\gamma'_i = 1$ and $\gamma_i = 0$ is given by
$$
\beta_{\bm{\gamma}\rightarrow \bm{\gamma}'} = p_{\bm{\gamma}\rightarrow \bm{\gamma}'} \frac{\pi(\bm{\theta}_{\bm{\gamma}} )}{\pi(\bm{\theta}_{\bm{\gamma'}})}\frac{\pi(\bm{\gamma})}{\pi(\bm{\gamma'})} = \frac{\pi(\bm{\theta}_{\bm{\gamma}})}{\pi(\bm{\theta}_{\bm{\gamma'}})}\frac{w}{(1-w)}
$$
Denoting $V_{\bm{\gamma}} = \bm{\Sigma}_{\bm{\gamma}}^{-1}$, the ratio simplifies as,
$$
\frac{\pi(\bm{\theta}_{\bm{\gamma}})}{\pi(\bm{\theta}_{\bm{\gamma'}})} = \frac{|V_{\bm{\gamma}}|^{1/2}\exp\left(-\frac{1}{2}(\bm{\theta}_{\bm{\gamma}} - \mu_{\bm{\gamma}})^TV_{\bm{\gamma}}(\bm{\theta}_{\bm{\gamma}} - \mu_{\bm{\gamma}})\right)}{\sqrt{2\pi}|V_{\bm{\gamma}'}|^{1/2}\exp\left(-\frac{1}{2}(\bm{\theta}_{\bm{\gamma}'}- \mu_{\bm{\gamma}'})^TV_{\bm{\gamma}'}(\bm{\theta}_{\bm{\gamma}'} - \mu_{\bm{\gamma}'})\right)}.
$$
In our examples the prior is independent across components and this ratio simplifies to a constant. As the prior mean is 0 and, if we denote the prior variance for $\theta_i$ for any active covariate $i$ as $\sigma^2$, we have
\[
\beta_{\bm{\gamma}\rightarrow \bm{\gamma}'} =p_{\bm{\gamma}\rightarrow \bm{\gamma}'}\frac{1}{\sqrt{2\pi\sigma^2}}\frac{w}{(1-w)}.
\]


\subsection{P\'{o}lya-Gamma Gibbs sampling for logistic regression}
\label{polyGam}
The Polya-Gamma Gibbs sampling approach is an auxiliary variable approach for Bayesian Logistic regression. A Polya-Gamma random variable $\omega \sim \mathcal{PG}(b, 0)$, $b > 0$, with probability density $p(\omega)$ has the property \citep{polson2013bayesian} that for any $\psi \in \mathbb{R}$ and $a \in \mathbb{R}$
$$
\frac{\exp(\psi)^a}{(1 + \exp(\psi))^b} = 2^{-b}\exp((a-\frac{b}{2})\psi)\int_0^{\infty}\exp(-\omega\frac{\psi^2}{2})p(\omega)\mbox{d}\omega.
$$
Thus the implied conditional distribution for $\psi$, given auxiliary variable $\omega$, is Gaussian. The advantage of this approach is that when updating the model $\bm{\gamma}$ we can integrate over the parameters $\bm{\theta}$ yielding much more efficient moves. The updates for the collapsed Gibbs sampling procedure follow the form:
\begin{enumerate}
    \item[(1)] sample $\bm{\gamma} \sim \bm{\gamma} \mid \bm{\omega}$;
    \item[(2)] sample $\bm{\theta} \sim \bm{\theta} \mid \bm{\omega}, \bm{\gamma}$;
    \item[(3)] sample $\bm{\omega} \sim \bm{\omega} \mid \bm{\theta}, \bm{\gamma}$.
\end{enumerate}
\textbf{Simulation step 1.}
\\
Let $\tilde{\pi}(\bm{\gamma} \mid \bm{\omega})$ be a density proportional to $\pi(\bm{\gamma} \mid \bm{\omega})$ such that
\begin{align*}
\tilde{\pi}(\bm{\gamma} \mid \bm{\omega}) &= \pi_0(\bm{\gamma})\int_{\bm{\theta}_{\bm{\gamma}}} L(y^{1:n}|\bm{\theta}_{\bm{\gamma}},\bm{\omega})\pi_0(\bm{\theta}_{\bm{\gamma}})\mbox{d}\bm{\theta}_{\bm{\gamma}}\\
&= \frac{\pi_0(\bm{\gamma})}{\sqrt{\text{det}(2\pi\sigma^2\bm{I}_{\bm{\gamma}})}}\int_{\bm{\theta}_{\bm{\gamma}}}\prod_{i=1}^n\exp\left((y_i - 0.5)(\bm{X}_{\bm{\gamma}}\bm{\theta}_{\bm{\gamma}})_i -\frac{\omega_i}{2}((\bm{X}_{\bm{\gamma}}\bm{\theta}_{\bm{\gamma}})_i)^2 - \frac{1}{2\sigma^2}\bm{\theta}_{\bm{\gamma}}^T\bm{\theta}_{\bm{\gamma}}\right) \mbox{d}\bm{\theta}_{\bm{\gamma}}\\
&= \pi_0(\bm{\gamma})\sqrt{\frac{\text{det}(2\pi \bm{V}_{\bm{\gamma}})}{\text{det}(2\pi\sigma^2\bm{I}_{\bm{\gamma}})}}\exp(\frac{1}{2}\kappa^T\bm{X}_{\bm{\gamma}}\bm{V}_{\bm{\gamma}}\bm{X}_{\bm{\gamma}}^T\kappa)
\end{align*}
where $(\bm{X}_{\bm{\gamma}}\bm{\theta}_{\bm{\gamma}})_i$ denotes the $i$th element of $\bm{X}_{\bm{\gamma}}\bm{\theta}_{\bm{\gamma}}$, the matrix $\bm{V}_{\bm{\gamma}} = \left(\bm{X}_{\bm{\gamma}}^T\Omega\bm{X}_{\bm{\gamma}} + \frac{1}{\sigma^2}\bm{I}_{\bm{\gamma}}\right)^{-1}$, the column vector $\kappa = y^{1:n} - 0.5$ and $\Omega = \text{diag}(\omega_1, ..., \omega_n)$. 
\\
The update for $\bm{\gamma}$ is taken by updating component-wise from the conditionals $\gamma_j \mid \bm{\gamma}_{(-j)}, \bm{\omega}$ where $\bm{\gamma}_{(-j)} = (\gamma_1, ..., \gamma_{j-1}, \gamma_{j+1}, ..., \gamma_p)$. Such a proposal can be implemented using the relationship \citep{Chipman2001}
$$
Pr(\gamma_j = 1 \mid \bm{\gamma}_{(-j)}, \bm{\omega}) = \frac{\tilde{\pi}(\gamma_j=1  \mid \bm{\gamma}_{(-j)}, \bm{\omega})}{\tilde{\pi}(\gamma_j=0  \mid \bm{\gamma}_{(-j)}, \bm{\omega})}\left(1 + \frac{\tilde{\pi}(\gamma_j=1  \mid \bm{\gamma}_{(-j)}, \bm{\omega})}{\tilde{\pi}(\gamma_j=0  \mid \bm{\gamma}_{(-j)}, \bm{\omega})}\right)^{-1}
$$
\\
\textbf{Simulation step 2.}
\\
The conditional for $\bm{\theta}_{\bm{\gamma}}$ is 
$$
\bm{\theta}_{\bm{\gamma}} | \bm{\omega}, \bm{\gamma}  \sim \mathcal{N}(\bm{m}_{\bm{\gamma}}, \bm{\Sigma}_{\gamma})
$$
where $\bm{\Sigma}_{\bm{\gamma}} = (\bm{X}_{\bm{\gamma}}^T\Omega\bm{X}_{\bm{\gamma}} + \frac{1}{\sigma^2}\bm{I}_{\bm{\gamma}})$ and $\bm{m}_{\bm{\gamma}} = \bm{\Sigma}_{\bm{\gamma}}\bm{X}_{\bm{\gamma}}\kappa$.
\\
\\
\textbf{Simulation step 3.}
\\
The conditional for $\bm{\omega}$ is $
\omega_i | \bm{\theta}, \bm{\gamma} \sim \mathcal{PG}\left(0, (\bm{X}_{\bm{\gamma}}\bm{\theta}_{\bm{\gamma}})_i\right)
$ for $i = 1,..., n$.

\section{Computation of relative efficiencies in Section 5.1}
\label{supp:RE}
In order to compute the relative efficiency we need an estimate of the statistical efficiency \eqref{eq:mse}. We estimate this quantity using a reference estimate $q$ from an independent 6-hour run of the Gibbs sampling method for each combination of $n$, $p$ and Scenario in Tables 1-3 and the results of Figure \ref{fig:my_label1}. For the results in Tables 1-3 the quantities of interest are the estimation of the posterior marginal inclusion probabilities $\pi(|\theta_j| > 0)$ (PPI) and the marginal posterior means $\mathbb{E}[\theta_j]$ (Mean). These two quantities allowed us to see how efficient the sampler was in terms of exploring both the parameter and model space. For the simulations in the subsampling comparison (Figure \ref{fig:my_label1}) the quantity of interest was the posterior mean conditioned on being in model $\bm{\gamma} = (1,1,0,0,...,0))$. We estimate the mean square error of these terms by running 100 independent runs of each algorithm and comparing to the corresponding long Gibbs run. Methods used in Tables 1-3 were initialised at the maximum a posterior estimate for a model with active variables chosen from an initial LASSO fit. For the subsampling comparison, methods were initialised at the location of the control variate (the maximum a posterior estimate with using the true nonzero variables). For each algorithm in the simulations of Tables 1-3 we use a computational budget of $10^6$ iterations with a maximum run time of 2 minutes. For the simulations in the subsampling comparison we used a computational budget of $10^6$ iterations with a maximum run time of 6 minutes. Algorithms were then compared on the basis of relative computational efficiency using RE or relative efficiency per iteration using RSE. An iteration for the Gibbs sampler is considered to be a full update of all parameters (i.e. one run of all steps in Section \ref{polyGam}) whereas an iteration of the PDMP methods is considered to be one simulated event time.

\end{document}